\documentclass{article}
\usepackage{amsmath, amssymb, amsthm, bbm}
\usepackage{enumitem}
\usepackage{tikz}
\usepackage[UKenglish]{isodate}
\usepackage[font=small,labelfont=it]{caption}
\usepackage{subcaption}
\usepackage{hyperref}
\usepackage{multirow}
\usepackage{fancyvrb}

\cleanlookdateon

\DeclareMathOperator{\Leb}{Leb}
\DeclareMathOperator{\supp}{supp}
\DeclareMathOperator{\inter}{int}
\DeclareMathOperator*{\esssup}{ess\,sup}
\newcommand{\n}{\mathbbm{n}}
\newcommand{\cL}{\mathcal{L}}
\newcommand{\cM}{\mathcal{M}}
\newcommand{\1}{\mathbbm{1}}
\newcommand{\II}{I_{\varepsilon,\Gamma}}

\makeatletter
\newcommand{\optionalitemlabel}[2]{%
  \phantomsection
  #1\protected@edef\@currentlabel{#1}\label{#2}%
}
\makeatother

\newtheorem*{theorem*}{Theorem}
\newtheorem{theorem}{Theorem}[section]
\newtheorem{atheorem}{Theorem}
\newtheorem{proposition}[theorem]{Proposition}
\newtheorem{corollary}[theorem]{Corollary}
\newtheorem{remark}[theorem]{Remark}
\newtheorem{lemma}[theorem]{Lemma}

\title{Analysis of bank leverage via dynamical systems and deep neural networks}

\author{Fabrizio~Lillo\thanks{Dipartimento di Matematica, Universit\'a di Bologna and Scuola Normale Superiore, Pisa, Italy. Email address: fabrizio.lillo@unibo.it} 
\and Giulia~Livieri\thanks{Scuola Normale Superiore, Pisa, Italy. Email address: giulia.livieri@sns.it} 
\and Stefano~Marmi\thanks{Scuola Normale Superiore, Pisa, Italy. Email address: stefano.marmi@sns.it} 
\and Anton~Solomko\thanks{Scuola Normale Superiore, Pisa, Italy. Email address: solomko.anton@gmail.com}
\and Sandro~Vaienti\thanks{Aix Marseille Universit\'e, Universit\'e de Toulon, CNRS, CPT, 13009 Marseille, France. Email address: vaienti@cpt.univ-mrs.fr} }

\date{\today}

\begin{document}

\maketitle

\begin{abstract}
We consider a model of a simple financial system consisting of a leveraged investor that invests in a risky asset and manages risk by using Value-at-Risk (VaR). 
The VaR is estimated by using past data via an adaptive expectation scheme. 
We show that the leverage dynamics can be described by a dynamical system of slow-fast type associated with a unimodal map on $[0,1]$ with an additive heteroscedastic noise whose variance is related to the portfolio rebalancing frequency to target leverage. 
In absence of noise the model is purely deterministic and the parameter space splits in two regions: 
(i) a region with a globally attracting fixed point or a 2-cycle; 
(ii) a dynamical core region, where the map could exhibit chaotic behavior.
Whenever the model is randomly perturbed, we prove the existence of a unique stationary density with bounded variation, the stochastic stability of the process and the almost certain existence and continuity of the Lyapunov exponent for the stationary measure.
We then use deep neural networks to estimate map parameters from a short time series.
Using this method, we estimate the model in a large dataset of US commercial banks over the period 2001--2014. 
We find that the parameters of a substantial fraction of banks lie in the dynamical core, and their leverage time series are consistent with a chaotic behavior.  
We also present evidence that the time series of the leverage of large banks tend to exhibit chaoticity more frequently than those of small banks.
\end{abstract}

{\small

{\bf 2020 Mathematics Subject Classiﬁcation:} primary 91G80, 34F05; secondary 37H15, 62M45.
}

\tableofcontents

\section{Introduction}
Leverage is one of the most important and controversial concepts in finance.
On one side, borrowing is essential in many economic activities, while, on the other, it is intrinsically connected with risk.
Concerning this last point, the pro-cyclical nature of leverage has been highlighted and studied in the recent literature (see, e.g.,\ \cite{Fostel08leveragecycles, adrianshin2010, huizingalaeven2012, Adrian13procyclicalleverage, acharyaryan2016, galothomas2017}, among many others).
Specifically, \cite{adrianshin2010} and \cite{Adrian13procyclicalleverage} argued that when assets are evaluated at \textit{fair value}, an increase in market prices of assets decreases the so-called ``quasi-market leverage ratio" -- roughly the ratio of total assets to equity capital -- and this leaves room to build up debt for banks that operate through leverage or Value at Risk\footnote{The Value at Risk measures the maximum loss that an asset portfolio may suffer over a specific horizon and with a given level of confidence.} (VaR) targeting.
The empirical findings in \cite{adrianshin2010} and \cite{Adrian13procyclicalleverage} show that cycles of expansion (contraction) in the banks' balance sheet size go hand in hand with increases (decreases) in leverage (i.e.,\ leverage is pro-cyclical); a behavior that has been witnessed since the 1960s but exacerbated during 2007--2009 financial crisis.
The creation of negative externalities in financial markets because VaR model's widespread use has been put forward in e.g., \cite{persaud2000, danielsson2011, barbara2017}.
It is shown that it can create market instability and result in what has been called by \cite{DANIELSSON20041069} as \textit{endogenous risk}, that is, the systemic risk caused and amplified by the system itself rather than being the result of an exogenous shock.
Indeed, because of the imposed VaR capital requirements, banks are forced to reduce their positions when the risk exceeds these limits.
Since the VaR of the trading portfolio increases when the volatility goes up,  banks are forced to reduce their positions rapidly, and because of these fire sales, the price can drop abruptly.
This leads to a new drop in prices and likely increases volatility, which triggers (further) portfolio reductions.
This mechanism creates an exceptionally threatening environment if many banks hold similar positions and use the same VaR model to manage their risk since they are forced to sell the same assets contemporaneously, leading to a destabilizing spiral.
Finally, to implement the VaR constraint (as well as any risk management mechanism), banks must estimate both the riskiness of the investments and their dependencies.
Since the estimation of risk is typically done using historical data, additional feedback is created between past and future risks,  creating new threats for the systemic stability of financial markets.

Understanding and modeling the leverage dynamics is therefore of paramount importance.
Some recent papers \cite{aymanns2015dynamics,corsi2016micro,mazzarisi2019panic} have proposed stylized agent-based models of financial institutions that lead to a dynamical system evolution for the leverage.
By using \textit{numerical methods}, they also show that in some parameters regime the leverage dynamics becomes chaotic via a cascade of period doubling bifurcations\footnote{Other works that analyze systemic risk problems through tools from dynamical system theory are, e.g. \cite{choi2012financial, choi2013financial, poledna2014leverage, castellacci2015modeling, aymanns2016taming}, to cite only a few.}.
The map describing some of these models also contains a noise component, describing the uncertainty naturally present in financial institutions' decision process.
For example, optimal leverage depends on portfolio risk, which is typically estimated with statistical methods on past observations and these estimations are naturally modeled as random variables.
In this paper, we address three important questions related to this approach (to be discussed below): 
(1) What are the mathematical properties (existence and uniqueness of stationary measure, stochastic stability, Lyapunov exponent, etc.) of the noisy deterministic maps emerging from these models? 
(2) Is it possible to reliably estimate the map parameters from short time series as those in publicly available datasets of banks balance sheets? 
(3) Is there evidence of chaotic behavior in the leverage dynamics of real banks?

We consider a simple agent-based financial system where the mechanisms described above are present.
Our starting point is a simplified version of the model proposed in \cite{mazzarisi2019panic}, which in turn builds on \cite{corsi2016micro}, to the case of one bank and one asset.
As \cite{mazzarisi2019panic} shows numerically, the dynamics describing the financial leverage displays a period-doubling bifurcation cascade resulting in chaotic behaviour (as measured with the computation of the Lyapunov exponent).
This happens  either when the parameter of the VaR constraint or the memory used by the banks to estimate volatility from past data vary.
These findings, while suggestive, are however not rigorously proved in \cite{mazzarisi2019panic}.
More importantly, the mathematical properties of this family of models, i.e.\ where a deterministic map is perturbed by a heteroscedastic additive noise (arising from the coupling with a faster random dynamics), are not known in general.
In the simplified model, the system is composed of a representative leveraged investor (a bank) that invests in a risky asset; the bank's risk management consists of two components.
First, the bank estimates the future volatility (the risk) of its investment in the risky asset using past market data.
Second, the bank uses the estimated volatility to set its desired leverage.
However, the bank faces a Value-at-Risk (VaR) capital requirement policy which implies a constraint for the financial leverage $\lambda_t := {A_t}/{E_t}$, where $A_t$ indicates the assets of the bank at time $t$, whereas $E_t$ the liabilities.
The bank is allowed a maximum leverage $\Bar{\lambda}_t$ which is a function of its own perceived risk.
The representative bank updates its expectation of risk at time intervals of unitary length, say $(t, t+1]$ with $t \in \mathbbm{Z}$, and, accordingly, it makes new decisions about the leverage.
This process defines the \textit{slow} component of the model.
Moreover, the model assumes that over the unitary time interval $(t, t+1]$ representative bank re-balances its portfolio to target the leverage without changing the risk expectations.
The re-balancing takes place in $\n$ time sub-intervals within $(t, t+1]$.
The time scale $1/\n$, with $\n \in \mathbbm{N}$, characterizes the \textit{fast} component of the model.
In particular, the slow variables evolve in time as a function of averages over the fast variables.
In summary, the considered model is a discrete-time slow-fast dynamical system (\cite{dolgopyat2005averaging} and \cite{berglund2006noise}).
Starting from this model, in this paper we make three main contributions.

First, we show that the dynamics of leverage in our model follows, under suitable approximation, a deterministic unimodal map on $[0,1]$ perturbed with additive and heteroscedastic noise.
The variance of the noise is related to the frequency of portfolio rebalancing to target leverage.
In particular, the parameter space of this deterministic map has two regions: 
(i) a region where the map has a globally attracting fixed point or a 2-cycle; 
(ii) the so-called \textit{dynamical core} region, where the map can exhibit chaotic behavior.
In order to study the mathematical properties of the map rigorously, we consider a more general class of maps and describe the leverage dynamic utilizing a Markov chain parametrized by the \textit{rebalance time} $\n$; we will study the regime of finite $\n$, as well as the limit for $\n\to\infty$.
Although the stability of Markov chains is relatively well studied (see e.g.,\ \cite{borovkov1998ergodicity} or \cite{meyn1993markov} and references therein), some specific properties of the stochastic kernel that defines our model do not allow us to apply general results available; for instance, we do not know if our chain is Harris recurrent.
We instead exploit the unimodal dynamics of the deterministic map.
Perturbations of unimodal maps with uniform additive noise were studied in \cite{baladi2002almost}, \cite{baladi1996strong}.
As far as we know, the Markov chains with the  kind of heteroscedastic noise we introduce are apparently new; see \cite{gianetto2012estimating} for another type of heteroscedastic nonlinear autoregressive process applied to financial time series.
To handle with them, we look at the Markov operator's spectral properties on suitable Banach spaces and prove the quasi-compactness of such an operator.
This allows us to obtain many rigorous results.
In particular, we get finitely many stationary measures with bounded variation densities.
The stationary measure's uniqueness is achieved when the chain perturbs the unimodal map which is either topologically transitive or admits an attracting periodic orbit; such maps correspond to a major part of the parameter space.
From a financial point of view, should the stationary measure not be unique, it would imply that, depending on the initial conditions, different banks could experience completely different dynamics, corresponding to different stationary measures.
When this occurs in physical systems one speaks of phase transitions and of coexistence of different mutually singular states.
In our case, this would imply, for instance, that policy measures could not be universal.
We also show the weak convergence of the unique stationary measure to the invariant measure of the unimodal map.
We point out that this step is particularly delicate since the stochastic kernel becomes singular in the limit of large $\n$.
It is well known that given a continuous Markov chain which perturbs a given map $T$ as in our setting, one could construct a sequence of random transformations close to $T$ and therefore replace the deterministic orbit of $T$ with a random orbit given by the concatenation of the maps randomly chosen in the sequence.
This construction is formally possible under general assumptions, but it is challenging to get ``representations by special classes of transformations" as Y.~Kifer pointed out in \cite{kifier1986ergodic}. 
We provide an explicit construction of those transformations and show their closeness with the unimodal map; we believe this inference from the Markov chain to the random transformations is interesting and illustrates very well the way the Markov chain moves randomly the states of the system.
In particular, our random maps are obtained by adding to the unimodal map $T$ an additive term which also depends on the state variable, and this motivates the attribute of heteroscedastic we gave to our noise.
Once we dispose of the stationary measure, we can define an average Lyapunov exponent by integrating the logarithm of the derivative of the unimodal map $T$ with respect to such a measure; this definition is suitable for the Markov chain approach.
By switching to random transformations, one can define the Lyapunov exponent of the cocycle.
We show numerically that the two approaches are asymptotically equivalent and prove that the average exponent converges to the Lyapunov exponent of $T$ in the limit of large $\n$, as a consequence of the stochastic stability.
We finally show that the average Lyapunov exponent depends continuously on the Markov chain parameters, and relate it to the different chaotic behavior of the unperturbed unimodal map.

One of the purposes of the present work was to rigorously establish the possibility of chaotic behaviour in leverage time series of banks, as well as to detect it in financial datasets.
For this reason we started from the parametric slow-fast model of \cite{mazzarisi2019panic} and looked at the corresponding Lyapunov exponent.
We remind that for deterministic systems, the  Lyapunov exponents  characterize the divergence of nearby orbits, allowing to distinguish between regular and chaotic dynamics.
In the presence of the noise, our model becomes intrinsically stochastic, in particular the Markov chain will mix exponentially fast, see section \ref{RLESection}. 
Nevertheless the average Lyapunov exponent still allows us to discriminate periodic and chaotic behaviours: it is negative when we perturb a contracting map (and then the realizations of the process fluctuate around the fixed point), 
and it becomes positive by perturbing the dynamical core region.
In both cases the distribution of the realizations in the state space is governed by the unique absolutely continuous stationary measure.
We prove that both behaviours, chaotic and periodic, are present in the parameter space and we actually observe both regimes in the real data.
This justifies our claims about chaotic dynamics of the leverage.

Indeed, in our third contribution (to be detailed below) we will detect possible chaos in banks' leverage dynamics.
As said, for the stochastic stability and the Lyapunov exponent we mainly discuss the relation between finite $\n$ (thus noisy system) and infinite $\n$ (purely deterministic system).
Thus our results indicate in which sense what we learn for a noisy system is informative about the deterministic backbone.
In our empirical analysis we will do not study or use directly neither the stationary measure nor the Lyapunov exponent (mainly because we have very short time series), however the "continuity" we observe from finite to infinite $\n$ suggests that the properties we observe empirically for finite $\n$ are informative of the underlying deterministic dynamics. 

The paper's second contribution concerns a methodology to infer the parameters of the noisy map from short empirical time series.
It is indeed interesting to ask if there is evidence of chaotic behavior in the leverage dynamics of real banks.
We claim that applying the maximum likelihood estimation is not feasible for two reasons.
First, the likelihood function is highly non-convex so that standard optimization methods may perform poorly.
Second, although the likelihood function for the process itself can be written explicitly, it may happen that in many cases we observe only a certain iterate of the process, e.g.,\ we observe only one slow time scale portfolio decision event out of two.
Therefore we propose to use a powerful deep learning technique known as Convolutional Neural Network (CNN) (\cite{lecun99}) to estimate the parameters of the map.
More precisely, our CNN takes as input the one-dimensional time-series and gives as output the map's corresponding parameters.
We train the CNN via extensive simulations of the model, considering different regions of the parameter space.
The robustness of the trained model is validated by its prediction of parameters performance in a huge testing set.
Results show the merit of using our proposed CNN architecture to estimate the parameters.
Importantly, being based only on simulations, the NN-based approach can also work for partial observations.
Without attempting to review the literature of parameter estimation of the dynamical system via NN, we only provide here a few key references to position our contribution. \cite{kumpati1990identification} employed multi-layer NN and recurrent networks to identify and control nonlinear deterministic dynamical systems.
Artificial NN have been used in \cite{materka1994application} and \cite{raol1996neural} in a similar framework.
Batch and recursive prediction error estimation algorithm have been derived for a NN model with a single hidden layer in \cite{chen1990non} and \cite{chen1992neural} for the identification of noisy discrete-time nonlinear dynamical systems.

The third contribution is the empirical analysis of real banks leverage time series.
Assuming the proposed unimodal map with heteroscedastic noise as data generating process for the banks' leverage, we estimate the parameters on quarterly data of about $5,000$ US Commercial  Banks provided by the Federal  Financial  Institutions  Examination  Council (FFIEC) via the proposed CNN architecture.
We have at our disposal a time period going from March 2001 to December 2014, for a total of 59 quarters.
Remarkably, we find that the parameters of a sizable fraction of banks lie in the map's dynamical core and that the large banks' leverage tends to be more chaotic than one of the small ones.
As a robustness check, identifying chaotic/periodic behavior is tested by following a non-parametric approach.
In this latter case, the map is not specified and estimators of indicators (such as the Lyapunov exponent \cite{zeng1992extracting}), which assume different values in the two regimes, are used to discriminate them from a finite length time series.
We use a very recent algorithm dubbed Chaos Decision Tree Algorithm \cite{toker2020simple} which combines several tools into an automated processing pipeline that can detect the presence (or absence) of chaos in noisy recordings, even for difficult edge cases.
We apply the Chaos Decision Tree Algorithm to our data set.
Remarkably, the results corroborate the CNN approach's findings concerning the chaotic behavior for a significant subset of typically large banks.

\medskip
\textbf{Outline of the paper.}
In Section~\ref{SMtoDSSection} we present the financial model of a representative bank managing its leverage.
We show that the model leads to a slow-fast deterministic-random dynamical system which can be recasted into a unimodal deterministic map with heteroscedastic noise.
In order to analyze it rigorously, in Section~\ref{ModelSction} we recall some facts about unimodal maps and Markov chains and then define the class of chains that we study.
We also represent our model in terms of random transformations.
In Section~\ref{StationaryMeasureSection} and \ref{StochasticStabilitySection} we show the existence and uniqueness of an absolutely continuous stationary measure and establish its convergence to the invariant measure of the deterministic map.
This allows us to define the Lyapunov exponent and prove its continuity with respect to the model parameters in Section~\ref{LyapunovExponentSection}.
We also discuss chaotic indicators naturally arising from the random maps representation of the process.
The last part of the paper concerns numerical and empirical analyses.
Specifically, Section~\ref{NumericalSection} presents some numerical investigations of the bifurcation diagram and Lyapunov exponent of the map.
Section~\ref{sec:estimationDNN} proposes an estimation method of the map based on the use of deep neural networks and Section~\ref{sec:empirical} presents an empirical application to a large set of leverage time series of US banks, showing evidence of chaotic behavior.
Finally, in Section~\ref{sec:conclusion} we draw some conclusions and outline some potential extensions of our work.

\section{From the structural model to the dynamical system} \label{SMtoDSSection}

The stylized model of the leverage dynamics we are going to present is a special case of the model of \cite{mazzarisi2019panic} (which in turn builds on \cite{corsi2016micro}) restricted to the case of a single (representative) financial institution and of a single investment asset.
We present below this model and show how, under suitable approximation, the resulting dynamics of leverage follows a deterministic map with additive and heteroscedastic noise.
The mathematical properties of such map are studied in the next sections.

In the model, a representative financial institution (hereafter a bank) takes investment decisions at discrete times $t \in {\mathbb Z}$ (the slow time scale).
At each time the bank's balance sheet is characterized by the asset $A_t$ and equity $E_t$, which together define the leverage $\lambda_t:= A_t/E_t$.
The bank wants to maximize leverage (by taking more debt) to increase profits, but regulation constrains the bank's Value-at-Risk (VaR) in such a way that $$
\lambda_t=\frac{1}{ \alpha \sigma_{e,t}},
$$
where $\alpha$ depends on the return distribution and VaR constraint\footnote{For example, if returns are Gaussian and the probability of VaR is 5\%, it is $\alpha=1.64$.}. $\sigma_{e,t}$ is the expected volatility at time $t$ of the asset, which in this simple model is composed by a representative risky investment.
Thus at each time $t$ the bank recomputes $\sigma_{e,t}$ and chooses $\lambda_t$.
Then in the interval $[t,t+1]$ the bank trades the risky investment to keep the leverage close to the target $\lambda_t$.
The trading process occurs on the points of a grid obtained by subdividing $[t,t+1]$ in $\n$ subintervals of length $1/\n$ (the fast time scale).
The dynamics of the investment return can be written as
\begin{equation}\label{eq:riskyreturn}
    r_{t+k/\n} = \varepsilon_{t+k/\n} + e_{t+(k-1)/\n},\quad k=1, 2, \ldots, \n,
\end{equation}
where $\varepsilon_{t+k/\n}$ and $e_{t+(k-1)/\n}$ are, respectively, the exogenous and endogenous component of the return.
The former is a white noise term with variance $\sigma^2_\epsilon$, while the latter depends on the banks' demand for the risky investment in the previous step.
For each bank, the demand for the risky investment at time $t+k/\n$ is the difference between the target value of $A_t$ to reach $\lambda_t$ and its actual value.
Since the bank's asset is composed by the risky investment, an investment return $r_{t+k/\n}$ modifies $A_t$ and the bank trades at each grid point to reach the target leverage.
It is possible to show (see \cite{corsi2016micro,mazzarisi2019panic}) that to achieve this, at each time $t+k/\n$ the bank's demand for the risky investment is
$$
D_{t+k/\n}=(\lambda_t-1) A^*_{t+(k-1)/\n} r_{t+k/\n}, 
$$
where $A^*_{t+(k-1)}$ is the target asset size in the previous step.
If there are $M$ identical banks, the aggregated demand is $MD_{t+k/\n}$.
The endogenous component of returns $e_{t+k/\n}$ is determined by the aggregated demand by the equation
$$
e_{t+k/\n}=\frac{1}{\gamma}\frac{MD_{t+k/\n}}{C_{t+k/\n}},
$$
where 
$C_{t+k/\n}=MA^*_{t+(k-1)/\n}$ is a proxy of the market capitalization of the risky asset and
$\gamma$ is a parameter measuring the investment liquidity.
Using the above expression, it is 
$$
e_{t+k/\n}=\frac{\lambda_t-1}{\gamma} e_{t+(k-1)/\n} = \phi_t e_{t+(k-1)/\n}
$$
and thus in the period $[t, t+1]$ the return $r_{t+k/\n}$  follows an AR(1) process with autoregression parameter $\phi_t=(\lambda_t-1)/\gamma$ and idiosyncratic variance $\sigma^2_\epsilon$.

To close the model, we specify how the bank forms expectations $\sigma_{e,t}$ on future volatility at time $t$.
We assume that bank uses adaptive expectations, which implies that
$$
\sigma^2_{e,t}=\omega \sigma^2_{e,t-1}+(1-\omega)\hat \sigma^2_{e,t},
$$
where $\omega \in [0,1]$ is a parameter weighting between the expectation at $t-1$ and the estimation $\hat\sigma^2_{e,t}$ of volatility   obtained by the return data in $[t-1,t]$.
As done in practice, this is obtained by estimating the sample variance of the returns in $[t-1,t]$, i.e.
\begin{multline}
\hat\sigma^2_{e,t} = \widehat{\text{Var}}\left[\sum_{k=1}^{\n} r_{t-1+k/\n}\right] \\
= \left(1+2\frac{\hat \phi_{t-1}(1-\hat \phi_{t-1}^\n)}{1-\hat \phi_{t-1}}-2\frac{(\n\hat \phi_{t-1}-\n-1)\hat \phi_{t-1}^{\n+1}+\hat \phi_{t-1}}{\n(1-\hat \phi_{t-1})^2}\right) \frac{\n \hat \sigma_{\epsilon}^2}{1-\hat \phi^2_{t-1}},
\end{multline}
where the last expression gives the aggregated variance of an AR(1) process as a function of the AR estimated parameters $\hat \phi_{t-1}$ and $\hat \sigma^2_\epsilon$.
In the following we will assume that these are the MLE estimators.
We remind that when $\n$ is large, $\hat \phi_{t-1}$ is a Gaussian distributed variable with mean $\phi_{t-1}$ and variance $(1-\phi^2_{t-1})/\n$.

In conclusion, the leverage dynamics is described by the following equations:
\begin{equation}\label{eq:model_final_1}
    \begin{cases}
    & \lambda_t=\left(\omega\frac{1}{\lambda^2_{t-1}}+(1-\omega)\alpha^2 \widehat{\text{Var}}[\sum_{k = 1}^{\n} r_{t-1+k/\n}]\right)^{-1/2},\\
    & r_s = \phi_{t-1}r_{s-1/\n} + \epsilon_s,\qquad s=t-1+k/\n,\quad k=1,2,\ldots,\n,
    \end{cases}
\end{equation}
Since slow variables evolve depending on averages of the fast variables, the model is a {\em slow-fast deterministic-random dynamical system}.
By using the expression above for the variance, we can rewrite the equation for the slow component only as
$$
\lambda_t=\left(\omega\frac{1}{\lambda^2_{t-1}}+(1-\omega)\alpha^2 \hat \sigma^2_{e,t}\right)^{-1/2},
$$
where the estimator $\hat \sigma^2_{e,t}$ can be seen as a stochastic term depending on $\lambda_{t-1}$ and whose variance goes to zero when $\n\to\infty$.

If $\n$ is large, the above map reduces to
$$
\lambda_t=\left(\omega\frac{1}{\lambda^2_{t-1}}+\frac{(1-\omega)\alpha^2 \n \hat \sigma_{\epsilon}^2}{(1-\hat \phi_{t-1})^2}\right)^{-1/2},
$$
and using the relation $\phi_t=\frac{\lambda_t-1}{\gamma}$, the map becomes
\begin{equation}
\phi_t=-\frac{1}{\gamma}+\frac{1}{\gamma}\left(\frac{\omega}{(1+\gamma \phi_{t-1})^2}+\frac{(1-\omega)\alpha^2 \n \hat \sigma_{\epsilon}^2}{(1-\hat \phi_{t-1})^2}\right)^{-1/2}.
\end{equation}\label{mapnew2}

When changing $\n$ also $\sigma^2_\epsilon$ changes, since the AR(1) can be seen as the discretization of a continuous time stochastic process (namely an Ornstein-Uhlenbeck process).
A simple scaling argument shows that the quantity $\Sigma_{\epsilon}=\sigma^2_\epsilon \n$ is instead constant and independent from the discretization step $1/\n$. 
In the limit $\n\to \infty$, it is $\hat \phi_t \to \phi_t$, thus the above map becomes purely deterministic\footnote{This is the deterministic skeleton, whose properties are discussed in detail in \cite{mazzarisi2019panic}.}.
The map in this case has a fixed point $\phi^{*} = \frac{1-\alpha \sqrt{\Sigma_{\epsilon}}}{1+\alpha\gamma\sqrt{\Sigma_{\epsilon}}}$.
By replacing this condition in \eqref{mapnew2} and assuming that the risky asset is very liquid ($\gamma \gg 1$), the map becomes
\begin{equation}\label{mapfabrizio}
\phi_t\simeq\left(\frac{\omega}{\phi^2_{t-1}}+\left(\frac{1-\phi^*}{\phi^*}\right)^2\frac{1-\omega}{(1-\hat \phi_{t-1})^2}\right)^{-1/2}.
\end{equation}

Since in the large $\n$ limit the ML estimator $\hat\phi_{t-1}$ is a Gaussian variable with mean $\phi_{t-1}$ and variance $(1-\phi^2_{t-1})/\n$, we can write 
$$
\hat\phi_{t-1}=\phi_{t-1}+\eta_{t-1},
$$
where $\eta_{t-1} \sim {\mathcal N}(0,(1-\phi^2_{t-1})/\n)$.
If the noise $\eta_{t-1}$ is small (i.e.\ $\n$ is large), we can perform a series expansion, obtaining
\begin{equation}
\phi_t \simeq \frac{|\phi_{t-1}(1-\phi_{t-1})|}{
\sqrt{b\phi^2_{t-1} + \omega (1-\phi_{t-1})^2}}(1+\zeta_{t-1}),
\end{equation}
where $b=b(\omega , \phi^*)$ is given by \begin{equation}\label{b_omega_phistar}
    b=(1-\omega)\left(\frac{1-\phi^*}{\phi^*}\right)^2
\end{equation}
and the noise term
$$
\zeta_{t-1}=\frac{-b \phi^2_{t-1}}{(1-\phi_{t-1})(b\phi^2_{t-1} + \omega (1-\phi_{t-1})^2)} \eta_{t-1}.
$$

In this approximation, the map can be seen as deterministic map with additive, but heteroscedastic, noise
\begin{equation} \label{MCAddNoise}
\phi_{t+1}=T(\phi_t; \theta)+\sigma(\phi_t; \theta) \epsilon_t,
\end{equation}
where $\epsilon_t\sim {\cal N}(0,1)$ and $\theta$ is a vector of parameters.
In our setting $\theta=(b,\omega,\n)$ and the deterministic map $T$ does not depend on $\n$.
Specifically,
\begin{equation}\label{unimodal_map}
T(\phi_t;\theta) = \frac{|\phi_t(1-\phi_t)|}{\sqrt{b\phi^2_t + \omega(1-\phi_t)^2}}
\end{equation}
and 
\begin{equation}\label{noise_variance}
\sigma(\phi_t; \theta) = \frac{b \phi_t^3 \sqrt{1-\phi_t^2}}{\sqrt{\n}\bigl(b\phi^2_t + \omega (1-\phi_t)^2\bigr)^{3/2}}.
\end{equation}
The map $T$ is represented in Fig.~\ref{fig:bump}.
The term under square root is always positive, since so are $\omega$ and $b$.

In the next sections we develop a rigorous mathematical theory of additive maps with heteroscedastic noise as in \eqref{MCAddNoise}.
In particular, we will study the existence of a stationary measure, the stochastic stability, and the Lyapunov exponent for this class of models, having in mind our main example of \eqref{unimodal_map} and \eqref{noise_variance} coming from the financial application.
It is worth noting, however, that our results remain valid for any noise $\epsilon_t$ in \eqref{MCAddNoise}, not only Gaussian; see Section~\ref{Coupling_section}.
Then in Section~\ref{sec:estimationDNN} we will present an estimation method for the map and in Section~\ref{sec:empirical} we will present the empirical analysis on real leverage time series.

\begin{figure}[h]
\centering
{\includegraphics[width=0.45\textwidth]{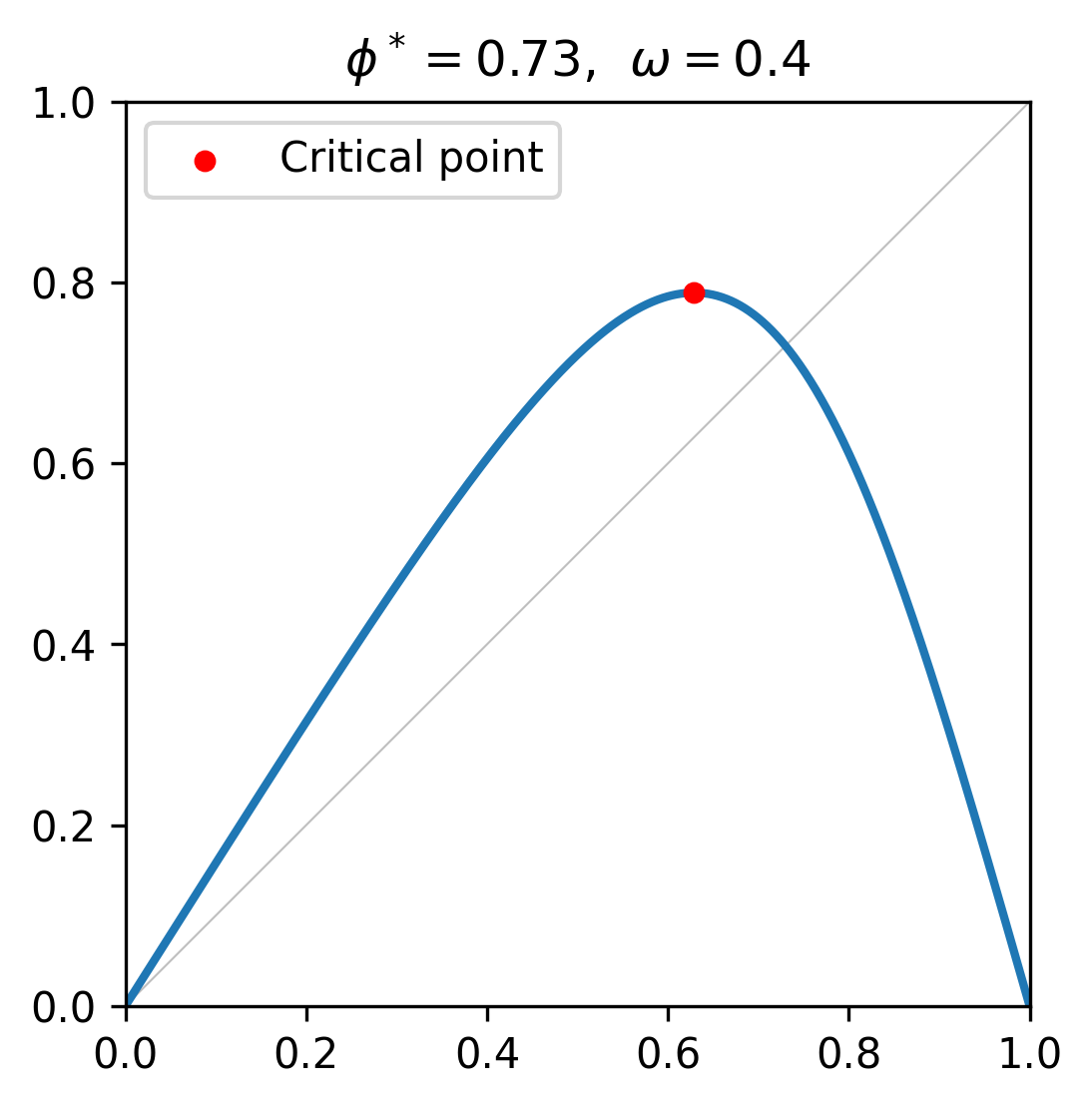}}
\caption{Plot of the deterministic component $T(\phi)$, $\phi^{*} = 0.73$, $\omega=0.4$ ($b=0.082$).}
\label{fig:bump}
\end{figure}

\section{The mathematical model}\label{ModelSction}

As we explained in the previous section, we model the dynamics of a financial system where two time scales are present:
a {\em slow time scale} where risk expectations (thus portfolios) are updated at unit times $t$, and a
{\em fast time scale}, between $t-1$ and $t$, during which banks rebalance portfolios $\n$ times.
Since slow variables evolve depending on averages of the fast variables, the model is a {\em slow-fast deterministic-random dynamical system}.
The evolution of the slow component is described by equations \eqref{mapfabrizio} and \eqref{MCAddNoise}.
Notice that since the distribution of $\phi_t$ only depends on $\phi_{t-1}$, both processes are (continuous state) Markov chains.
In this section, we define some tools that we later use to analyse the model.

\subsection{Unimodal maps}\label{UnimodalMapsSection}

Our process is constructed by perturbing with a heteroscedastic additive noise a deterministic unimodal map $T$ of the unit interval $I=[0,1]$.
We now describe the class of unimodal maps $T$ we will consider, with \eqref{unimodal_map} being its representative.

We will refer to the class of unimodal maps studied in \cite{baladi1996strong}; they enjoy a series of ergodic properties which allow us to establish rigorous results for the problem we deal with.  

We therefore consider unimodal maps $T\colon I\to I$ of class $C^4$ with $T(0)=T(1)=0$ and with a non-degenerate critical point\footnote{The critical point for \eqref{unimodal_map} is $c = \bigl(1 + \sqrt[3]{b/\omega}\bigr)^{-1}$.} at $c$: $T'(c)=0$.
The map $T$ is strictly increasing on $[0, c)$ and strictly decreasing on $(c, 1]$. 
Moreover, we suppose that $T$ satisfies the following assumptions:
\begin{enumerate}[label=(A\arabic*)]
\item\label{itm:A1} $T$ has negative Schwarzian derivative: $S(T) = \dfrac{T'''}{T'}-\dfrac32 \left(\dfrac{T''}{T'}\right)^2<0$,
\item\label{itm:A2} $\Delta:=T(c)<1$,
\item\label{itm:A3} the critical point is \emph{quadratic}: $T''(c)\ne 0$,
\item\label{itm:A4} $|T^k(c)-c|\ge e^{-\alpha k} $ for all $k\ge k_0$,
\item\label{itm:A5} $|(T^k)'(\Delta)|\ge \lambda_c^k$ for all $k\ge k_0$,
\end{enumerate}
where $k_0\ge 1, 1<\lambda_c<2$, and $\alpha>0$ with $e^{2\alpha}<\sqrt{\lambda_c}$ are fixed constants.

\begin{figure}[h]
\centering
{\includegraphics[width=0.5\textwidth]{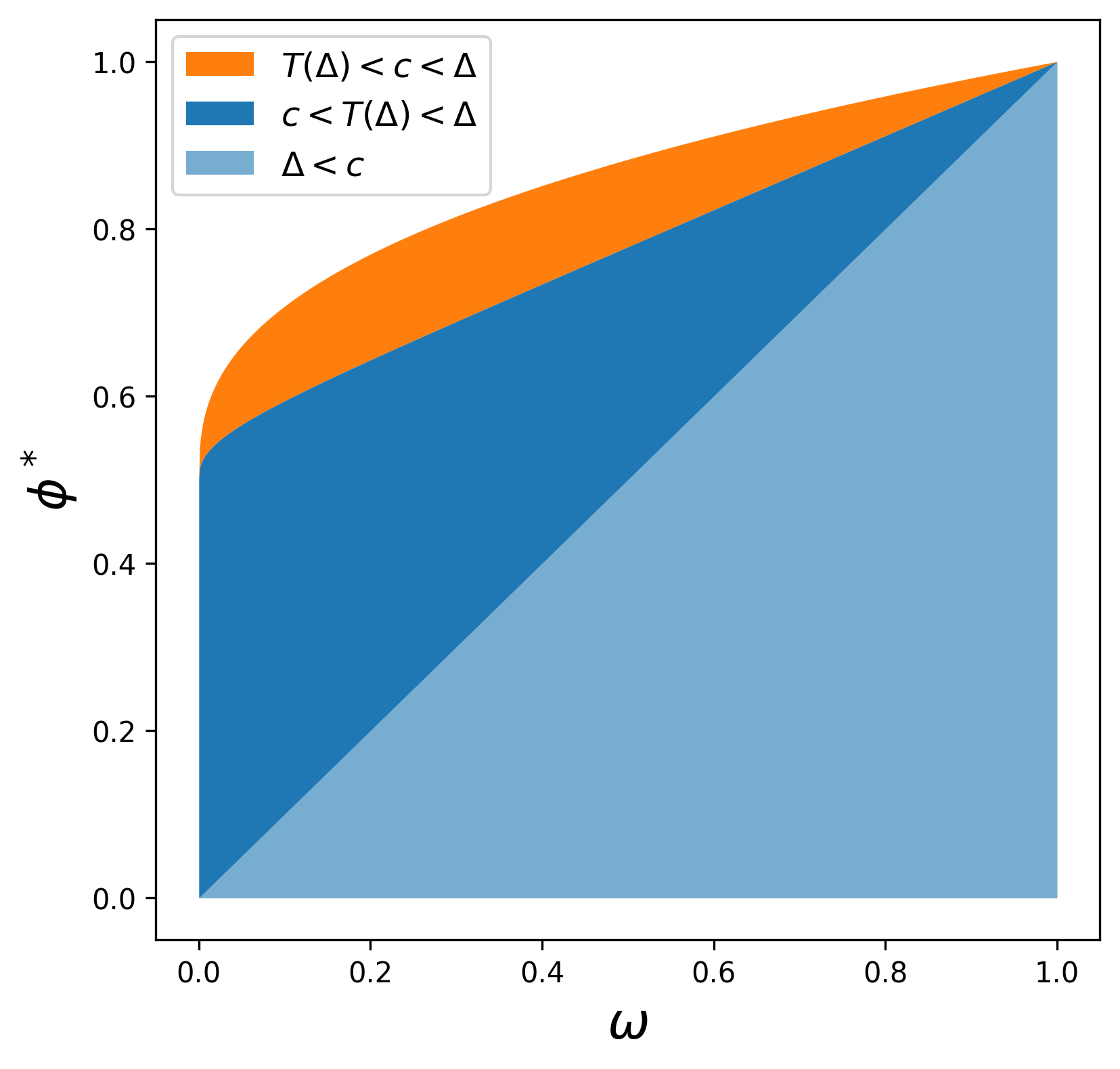}}
\caption{Partition of the parameter space for the unimodal map \eqref{unimodal_map} according to the classification \ref{itm:C1}--\ref{itm:C3}.}
\label{fig:dyncore}
\end{figure}

According to the theory of unimodal maps (see for instance the review \cite{thunberg2001periodicity}), we have the following classification: 
\begin{enumerate}[label=(C\arabic*)]
\item\label{itm:C1} If $\Delta\le c$, then there exists a globally attracting fixed point.
\item\label{itm:C2} If $c < T(\Delta) < \Delta$, then there is a globally attracting fixed point or a 2-cycle in $(c, \Delta)$.
\item\label{itm:C3} If $T(\Delta) < c < \Delta$, we can  reduce the study to the so-called \emph{dynamical core} $[T(\Delta), \Delta]$, which is mapped onto itself and absorbs all initial conditions (except $0$ which is a fixed point).
\end{enumerate}
Following this classification, we will say that $T$ is \textbf{periodic} if there is a globally attracting fixed point or a globally attracting cycle.
Recall that a map $T$ on a topological space $X$ is called \emph{topologically transitive} if  for all nonempty open sets $U, V\subset X$, there exists $n$ such that $T^{-n}U\cap V\neq \emptyset$.
We will say that $T$ is \textbf{chaotic} if, in addition to  above assumptions \ref{itm:A1}--\ref{itm:A5},
\begin{enumerate}
\item[\optionalitemlabel{(A$_t$)}{At}]
$T$ is topologically transitive on the interval $[T(\Delta), \Delta]$.
\end{enumerate}
In particular, any such $T$ satisfies the dynamical core condition \ref{itm:C3}.
Notice also that $T(x)>x$ for any $x < T(\Delta)$, in particular, $T'(0) > 1$.

The invariant sets (attractors) of a unimodal map have a variety of structures, as it is stated by the following theorem of Blokh and Lyubich (we quote the statement given in \cite[Theorem~6]{thunberg2001periodicity}):

\begin{atheorem}[\cite{blokh1991measurable}]\label{thmA}
\normalfont
Let $T: I\to I$ be an S-unimodal map with nonflat critical point (S means $S(T)<0$). 
Then $T$ has a unique metric attractor $A$, such that the $\omega$-limit set $\omega(x) = A$ for Lebesgue almost all $x\in I$.
The attractor $A$ is of one of the following types:
\begin{enumerate}
    \item an attracting periodic orbit;
    \item a Cantor set of measure zero;
    \item a finite union of intervals with a dense orbit.
\end{enumerate}
In the first two cases, $A=\omega(c)$.
\end{atheorem}

Associated to $T$ there is the \emph{transfer operator} (also called the \emph{Perron-Frobenius operator}) $L\colon L^1\to L^1$ which is the positive linear operator defined by the duality relation\footnote{Without mention of the contrary, all the $L^p$ spaces in the paper will be intended with respect to the Lebesgue measure.
The latter will be denoted as $dx$ or $\text{Leb}$.}
$$
\int_I Lf\, g = \int_I f\, g\circ T, \quad f\in L^1, g\in L^\infty.
$$

In order to get useful information from this operator, we need to restrict the functional space where it acts; we choose here the Banach space $BV$ of bounded variation functions on the unit interval equipped with the complete norm
$$
\|f\|_{BV}=|f|_{TV}+\|f\|_1,
$$
where $|f|_{TV}$ is the total variation of the function $f\in L^1$.
For a chaotic map $T$ it follows that it admits a unique absolutely continuous invariant measure $\nu = \nu\circ T^{-1}$ with a density $h\in BV$ and supported on the interval $[T(\Delta), \Delta]$ (\cite[Section~5, Corollary~1]{baladi1996strong}). 
Moreover, $\nu$ is mixing with exponential decay of correlations on BV observable, namely there exists $0<v<1$ and a constant $C>0$ such that
\begin{equation}\label{DE}
\Bigl| \int L^n f \, g dx-\int fdx \int gdx\Bigr|\le C v^n \|f\|_{BV}\|g\|_{\infty},
\end{equation}
see \cite[Section~5, Corollary~3]{baladi1996strong} and \cite[Proposition~5.15]{viana1997stochastic}.

\medskip

Before continuing, it is useful to quote a sort of analog of the theorem of Blokh and Lyubich given above, for what concerns invariant measures for the map $T$.
We give here the statement of Theorem~9 in \cite{thunberg2001periodicity}.

\begin{atheorem}[see {\cite[Chapter~V.1]{melo2012onedimensional}}] \label{thmB}
\normalfont
Let $T$ be an S-unimodal map with nonflat critical point.
If $T$ has a periodic attractor, or a Cantor attractor, then $T$ admits a unique SRB measure\footnote{We remind that an invariant measure $\mu$ is called a Sinai-Ruelle-Bowen (SRB) measure if
$$
\mu=\lim_{n\to \infty}\frac1n\sum_{k=0}^{n-1}\delta_{T^k(x)}
$$
for $\Leb$-a.e.\ $x\in[0,1]$, where $\delta_x$ is the Dirac mass at $x$.}
supported on the attractor.\\
If $T$ admits an absolutely continuous invariant probability measure $\mu$, then:
\begin{enumerate}
    \item\label{itm:BL1} $\mu$ is a SRB measure;
    \item\label{itm:BL2} the attractor $A$ of $T$ is an interval attractor;
    \item\label{itm:BL3} $\supp\mu = A$, in particular, $\mu$ is equivalent to the Lebesgue measure on $A$.
\end{enumerate}
\end{atheorem}

\begin{remark}
\normalfont
As pointed out in \cite{baladi1996strong}, \ref{itm:A4} ensures that $T$ has no periodic attractors and is ergodic with respect to the  measure $\mu$, which is the unique absolutely continuous invariant probability measure for $T$.
And \ref{At} allows to prove that $T$ is Bernoulli and therefore mixing. 
Therefore chaotic map fulfills conditions \ref{itm:BL1}--\ref{itm:BL3} of Theorem~\ref{thmB} above. 
As it is pointed out in \cite[Section~4]{thunberg2001periodicity}: ``The theorem [above] does not guarantee the existence of a natural
measure in the case of an interval attractor.
Indeed there are uncountably many parameters in the logistic family, for which the corresponding maps have interval attractors and lack natural measures altogether, or have natural measures with weird properties...
But, at least in the logistic family, both such singular phenomena and Cantor attractors are rare in the sense of Lebesgue measure''.
\end{remark}

In Section~\ref{LyapunovExponentSection} we will address the question to compute the Lyapunov exponent $\Lambda$ for the map $T$. 
The following theorem by G.~Keller can be applied to our situation:

\begin{atheorem}[\cite{keller1990exponents}] \label{thmC}
\normalfont
Let $T\colon I \to I$ be an S-unimodal map with nonflat
critical point.
Then $T$ admits an absolutely continuous invariant probability measure if and only if
$\lim_{n\to \infty}\frac{1}{n}\log|DT^n(x)|= \Lambda>0$, for almost all $x\in I$.
\end{atheorem}

Of course if  $T$ has a periodic attractor, the Lyapunov exponent will be negative.
The situation is different whenever $T$ has a Cantor attractor.
For an S-unimodal map with nondegenerate critical point that also has a Cantor
attractor, the Lyapunov exponent will be $0$, while there are families of unimodal maps with critical point of sufficiently high order, which have Cantor attractors with
sensitive dependence on initial conditions, see \cite[Section~5]{thunberg2001periodicity}.

\subsection{Markov chains}

Recall that a Markov chain $\{X_{t} : t\in\mathbb{N}\}$ on the interval $I$ is given by \emph{transition probabilities}
$$
P_x(A) = \mathbb{P}\{X_{t+1} \in A \mid X_t=x\}
$$
(the probability that a chain at $x\in I$ will be in a set\footnote{
All sets considered will be assumed to be measurable.
For brevity's sake, we omit the word `measurable' everywhere in this text.
}
$A\subset I$ after one step)
and an initial distribution $\rho_0(A) = \mathbb{P}\{X_{0} \in A\}$.
In the particular case where all $P_x$, $x\in I$, and $\rho_0$ are absolutely continuous (with respect to $\Leb$) and are given by densities $p(x,\cdot)$ and $h_0$, respectively, we have
$$
P_x(A) = \int_{A} p(x,y) dy, \quad \rho_0(A)=\int_A h_0(y) dy.
$$
The map $p\colon I\times I \to \mathbb{R}_+$ (known as the \emph{stochastic kernel}) plays the role that the transition matrix does in the theory of Markov chains with a finite state space.
For $P_x$ to be a probability, it should satisfy $\int p(x,y) dy = 1$ for every $x\in I$.

Denote with $\cM$ the space of (real-valued) Radon measures on $I$.
There is an associated operator $\cL\colon\cM\to\cM$ (called the \emph{Markov operator} corresponding to $P$) acting by
\begin{equation} \label{MO}
\cL \rho = \int P_x d\rho(x), \quad \rho \in \cM,
\end{equation}
that is,
$\cL \rho(A) = \int P_x(A) d\rho(x)$ for every $A\subset I$, or, equivalently, 
$$
\int \varphi d\cL\rho = \iint \varphi(y) dP_x(y) d\rho(x)
$$
for all $\varphi\in C^0$, where $C^0$ denotes the Banach space of continuous functions on $I$ with the $\sup$ norm.
We note that $\mathcal L\colon L^1 \to L^1$ is an isometry, where $L^1$ is intended, from now on, with respect to the Lebesgue measure.
If the chain is given by the kernel $p$, formula \eqref{MO} restricted to $L^1$ becomes
\begin{equation} \label{MO1}
(\cL h)(y) = \int p(x,y) h(x) dx, \quad h \in L^1.
\end{equation}

If $\rho_t$ denotes the distribution of the random variable $X_t$, then the distribution of $X_{t+1}$ is $\rho_{t+1} = \cL \rho_t$.
In other words, fixing the distribution $\rho_0$ for $X_0$, the entire sequence of future distributions can be obtained by iterating with $\cL$.

A measure $\mu\in\cM_+$ is said to be \emph{stationary} if
$$
\cL\mu = \mu.
$$
Every stationary measure $\mu$ gives rise to a shift-invariant measure $\mathbb{P}_\mu$ on the sequences space $\Omega=\{(x_t)_{t\in\mathbb{N}} : x_t\in I$\} of realizations of the process, such that $\mathbb{P}_\mu(x_t\in A)=\mu(A)$ for all $t\in\mathbb{N}$ (see, e.g., \cite{koralov2012theory}).
We say that $\mu$ is \emph{ergodic} if $\mathbb{P}_\mu$ is ergodic.
In the next section we will show that, under some mild conditions, our model admits a unique (and thus ergodic) absolutely continuous stationary probability $\mu$ with a density $h\in BV$.
Then, by the Ergodic Theorem (see, e.g., \cite[Remark~C4.1]{bhattacharya2007random}), for every $f \in L^1(\mu)$,
\begin{equation*}
\frac{1}{n} \sum_{t=1}^{n} f(X_t) \xrightarrow[n\to\infty]{} 
\int f d\mu, 
\quad \text{$\mathbb{P}_\mu$-almost surely}.
\end{equation*}
In particular, realizations of the process are distributed in the state space according to the measure $\mu$.

\subsection{Coupling with a stochastic process} \label{Coupling_section}

We are now ready to define a Markov chain that describes our model.
It is obtained as a deterministic unimodal map $T\colon I\to I$ satisfying assumptions of Section~\ref{UnimodalMapsSection}, coupled with a stochastic process, namely, by perturbing $T$ with additive noise.
Starting from \eqref{MCAddNoise} as our main motivation, on the one hand, we consider a more general class and, on the other, impose some mild technical restrictions that are necessary for rigorous analysis.

Since the noise varies in a neighborhood of 0, we will need to extend the state space on the negative axes.
We will see in a moment, however, that such an extension is irrelevant for the asymptotic behavior of the perturbed system, whose random trajectories spend all the time, but a relatively short transient, on the positive unit interval.

We fix $T$ and parametrize the chain by the {\em rebalance time} $\n$ (which is roughly inverse to the variance of the noise), consequently indexing with $\n$ the chain $(X_t^{(\n)})$, the transition probabilities $P_x^{(\n)}$ and the stochastic kernel $p_{\n}(x,y)$ where it is necessary.
We will be interested in the limit for $\n\to\infty$.

We also need to assume that the noise is compactly supported, in order for trajectories of the process to stay bounded.
Compared to the Gaussian noise in \eqref{MCAddNoise}, this is done by truncating the distribution tails that are exponentially small for large $\n$, see Section~\ref{MainExampleSection} for the main example.

Denote by
$$
\Gamma:=1-\Delta
$$
the gap between $T(c)$ and $1$.
We now extend the domain of definition of $T$ to the larger interval $[-\Gamma, 1]$ (which, by abuse of notation, will be still denoted by $I$) so that $T$ is continuous at $0$ and on $[-\Gamma, 0)$ is $C^4$ smooth, positive and decreasing, with $T(-\Gamma)<\Delta$.\footnote{A similar extension was considered in \cite{baladi1996strong} to allow perturbations with additive noise; in particular, it was supposed that $T(I)\subset \inter(I)$ and that $T$ admits an extension to some compact interval $J\supset I$, preserving all the previous properties and satisfying $T(\partial J)\subset \partial J$.
Notice that with this extension the map $T$ is not anymore of class $C^4$ as prescribed at the beginning of Section~\ref{UnimodalMapsSection}, but this regularity still persist on the interval $(0,1)$ and this will be enough for the next considerations.}

To construct the chain, we need to define transition probabilities.
Let $g_{x,\n}(y)$ be a probability density supported on a compact interval $[-s(x), s(x)]$.
We assume that
\begin{itemize}
    \item $0<s(x)<\Gamma/2$ for $x\in(0,1)$,
    \item $T(x) - s(x) > 0$ for $x\in(0, 1-\Gamma/2]$,
    \item $T(x) - s(x) > x$ for $x$ small (in particular, $T'(0) > 1$).
\end{itemize}
We set for simplicity $s(x):=0$ for $x\leq 0$ (meaning $P_x^{(\n)} = \delta_{T(x)}$);
Lemma~\ref{support_lemma1} below shows that this choice does not affect the dynamics.
We will also assume that both the mean and the variance of $g_{x,\n}$ decrease to $0$ as $\n\to\infty$ and, for every $\varepsilon > 0$,
\begin{equation}\label{g_TV_bound}
    \sup_{x\in[\varepsilon, 1-\varepsilon]} |g_{x,\n}|_{TV} < \infty.
\end{equation}

Fix any initial distribution $\rho_0\in \text{BV}$ and define transition probabilities
\begin{equation}\label{Px}
P_x^{(\n)}(A)
:= \int_{-s(x)}^{s(x)} {\1}_A(T(x)+y) g_{x, \n}(y)dy,
\end{equation}
which correspond to the stochastic kernel
\begin{equation}\label{pn}
p_\n(x,y) = g_{x,\n}(y - T(x)).
\end{equation}
Informally speaking, the probability that the chain steps from $x$ to $A$ will be high whenever $T(x)$ falls in $A$.
Equivalently, we can write
\begin{equation}\label{Proc1}
X_{t+1}^{(\n)} = T( X_{t}^{(\n)})+Y_{t+1}, \;\text{ where }\; Y_{t+1} \sim g_{x,\n}.    
\end{equation}
The values of $X_{t+1}^{(\n)}$ are spread in a neighborhood of $T(x)$ due to the addition of the random variable $Y_{t+1}$.

\subsection{The leverage model}  \label{MainExampleSection}

We will slightly modify \eqref{MCAddNoise} to satisfy the technical assumptions listed above.
The unimodal map is
\begin{equation}\label{unimodal_map_rewritten}
T(x) = \frac{|x(1-x)|}{\sqrt{b x^2 + \omega(1-x)^2}},
\end{equation}
where the parameters $\phi^*$ and $\omega$ are such that $T$ satisfies assumptions \ref{itm:A1}--\ref{itm:A5}. (see Section~\ref{UnimodalMapsSection}).
It always has negative Schwarzian derivative, as we verified numerically.
The condition $T(\Delta) < c < \Delta < 1$ defines a nonempty subset of parameters (see Fig.~\ref{fig:dyncore}).
Notice that $T'(0) = 1/\sqrt{\omega} > 1$.

We want $g_{x,\n}$ to be (truncated) normal with a variance close to \eqref{noise_variance}.
For this, let us denote by $\mathcal{N}_a(0,\sigma)$ the smoothed truncated normal distribution with the density
$
g(y) = c_{a,\sigma} \chi_a(y) e^{-\frac{y^2}{2\sigma^2}},
$
where $c_{a,\sigma}$ is so that $\int g(y) dy = 1$ and $\chi_a$ is a smooth bump function supported on $[-a,a]$.\footnote{The smoothness of the  truncation function is only used in the proof of Theorem~\ref{Lexp_cont_thm}. }
For instance, we may set 
$$
\chi_a(y) =
\begin{cases}
1, &\mbox{if } |y|\le (1-\varepsilon)a,\\
\Psi(\frac{y\pm(1-\varepsilon)a}{\varepsilon a}), &\mbox{if } (1-\varepsilon)a<|y|\le a,\\
0, &\mbox{if } |y|>a,
\end{cases}
$$
where $\Psi(t) = e^{1-\frac{1}{1-t^2}}$ is the standard $C^\infty$ bump function on $[-1, 1]$.
We set
\begin{equation}\label{noise_variance_rewritten}
\sigma_\n(x) := \frac{b x^3 \sqrt{1-x^2}}{\sqrt{\n} \bigl(b x^2 + \omega (1-x)^2\bigr)^{3/2}}.
\end{equation}
Denote $\sigma(x):=\sigma_1(x)$ and $\sigma_{\max} := \max_{x\in[0,1]} \sigma(x)$.
Set
\begin{equation}\label{noise_density}
g_{x,\n}(y) := c_{x,\n} \chi_{s(x)}(y)  e^{-\frac{y^2}{2\sigma^2_\n(x)}},
\end{equation}
where
$$
s(x) := \frac{\sigma(x)}{\sigma_{\max}} \min\Bigl\{\frac{\Gamma}{2}, T(1-\frac{\Gamma}{2})\Bigr\}
\quad \text{and} \quad
c_{x,\n} := \biggl( 
\int_{-s(x)}^{s(x)} \chi_{s(x)}(y) e^{-\frac{y^2}{2\sigma^2_\n(x)}} dy 
\biggr)^{-1}.
$$
We can then rewrite \eqref{Proc1} as
\begin{equation}\label{Proc2}
X_{t+1}^{(\n)} = T( X_{t}^{(\n)}) + \sigma_\n(X_t)Z_{t+1}, \quad Z_{t+1}\sim \mathcal{N}_{b_\n}(0,1),   
\end{equation}
with $b_\n := \frac{\sqrt{\n}}{\sigma_{\max}} \min\{\frac{\Gamma}{2}, T(1-\frac{\Gamma}{2})\}\to\infty$ as $\n\to\infty$.
The fact that, for $\n$ fixed, $g_{x,\n}$ are rescaled copies of the same distribution will be used in Section~\ref{RandomTransformSection} to explicitly describe random maps associated to the process.
We get the following stochastic kernel:
\begin{equation}\label{main_ex_st_kernel}
p_\n(x,y)
= c_{x,\n} \chi_{s(x)}(y-T(x)) e^{-\frac{(y-T(x))^2}{2\sigma_\n^2(x)}}.
\end{equation}
Notice also that the support of $p_\n(x,y)$ does not depend on $\n$ (see Fig.~\ref{fig:p}).

\begin{figure}[h]
\centering
{\includegraphics[width=0.5\textwidth]{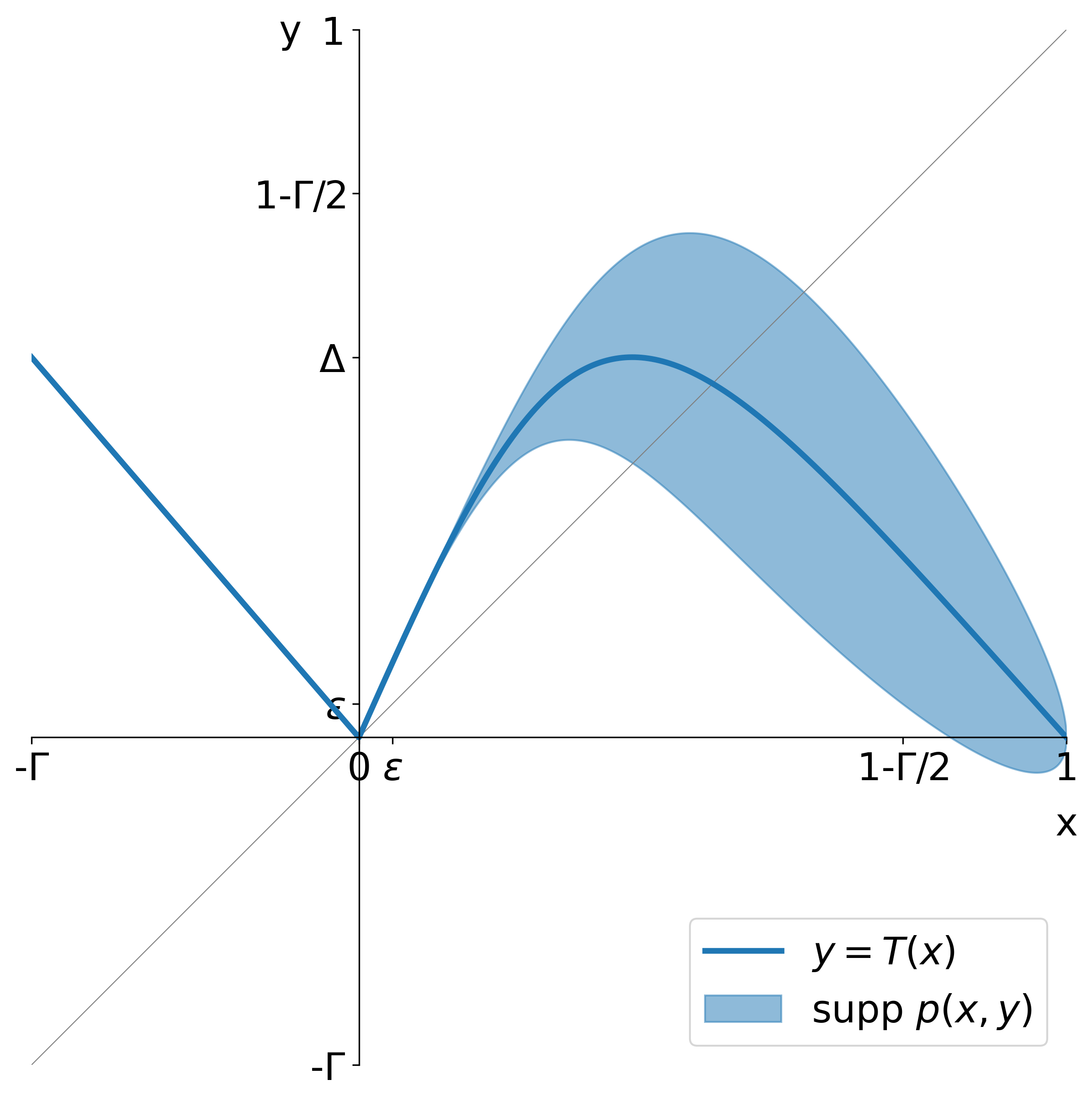}}
\caption{Support of the kernel \eqref{main_ex_st_kernel}.}
\label{fig:p}
\end{figure}

Finally, \eqref{g_TV_bound} holds, because $|g_{x,\n}|_{TV} = 2 c_{x,\n}$ and the latter is proportional to $1/\sigma_\n(x)$, which is bounded on $[\varepsilon, 1-\varepsilon]$.

\subsection{Random transformations}\label{RandomTransformSection}

Our model was defined as a Markov chain.
We now present a slightly different, yet equivalent, point of view.
Namely, we will pick up a family of maps $T_\eta\colon I\to I$, $\eta\in[0,1]$, in such a way that
\begin{equation} \label{RT2MC}
\Leb\{\eta : T_\eta(x)\in A\} = P_x(A)
\end{equation}
for all $A\subset I$.
We can then define a stochastic process
\begin{equation} \label{rmp}
\bar x_{t+1} = T_{\eta_{t+1}}(\bar x_{t}),
\end{equation}
where $\eta_t$ are independent and uniformly distributed in $[0,1]$.
We can write $\bar x_t = T_{\eta_t}\circ\cdots\circ T_{\eta_1} \bar x_0$, where $(\eta_t)_{t\in\mathbb{N}}$ is an i.i.d.\ stochastic process, 
i.e.\ the process \eqref{rmp} follows the orbits under the concatenation of randomly chosen maps from the family.
One can show that the two processes are equivalent, see, for instance, \cite{kifier1986ergodic}.
Conversely, starting with a family of maps $T_\eta$ one can use \eqref{RT2MC} to define transition probabilities $P_x$ and thus a Markov chain.

Rewriting \eqref{RT2MC} as
$
P_x(A) = \int_0^1 \1_A(T_\eta(x)) d\eta
$
and plugging into \eqref{MO} we get the \emph{disintegration formula} for the Markov operator:
\begin{equation}\label{dis}
\cL = \int L_\eta d\eta,
\end{equation}
where $L_\eta$ are the transfer operators associated with $T_\eta$.
In particular, a measure $\mu$ is stationary for the Markov chain if and only if it satisfies $\mu = \int L_\eta \mu d\eta$, i.e., for all $A\subset I$,
\begin{equation}\label{stationary_measure_random_maps}
\mu(A) = \int \mu(T_\eta^{-1}A) d\eta.
\end{equation}
Equation \eqref{stationary_measure_random_maps} is usually taken as the definition of a stationary measure for the family of random maps.
Every such measure corresponds to a product measure that is invariant for the skew-product with the Bernoulli shift in the base and the maps $T_\eta$ in the fibers;
we refer to \cite[Section~2]{aimino2015annealed} for details.

\begin{figure}[h]
\centering
\includegraphics[width=1\linewidth]{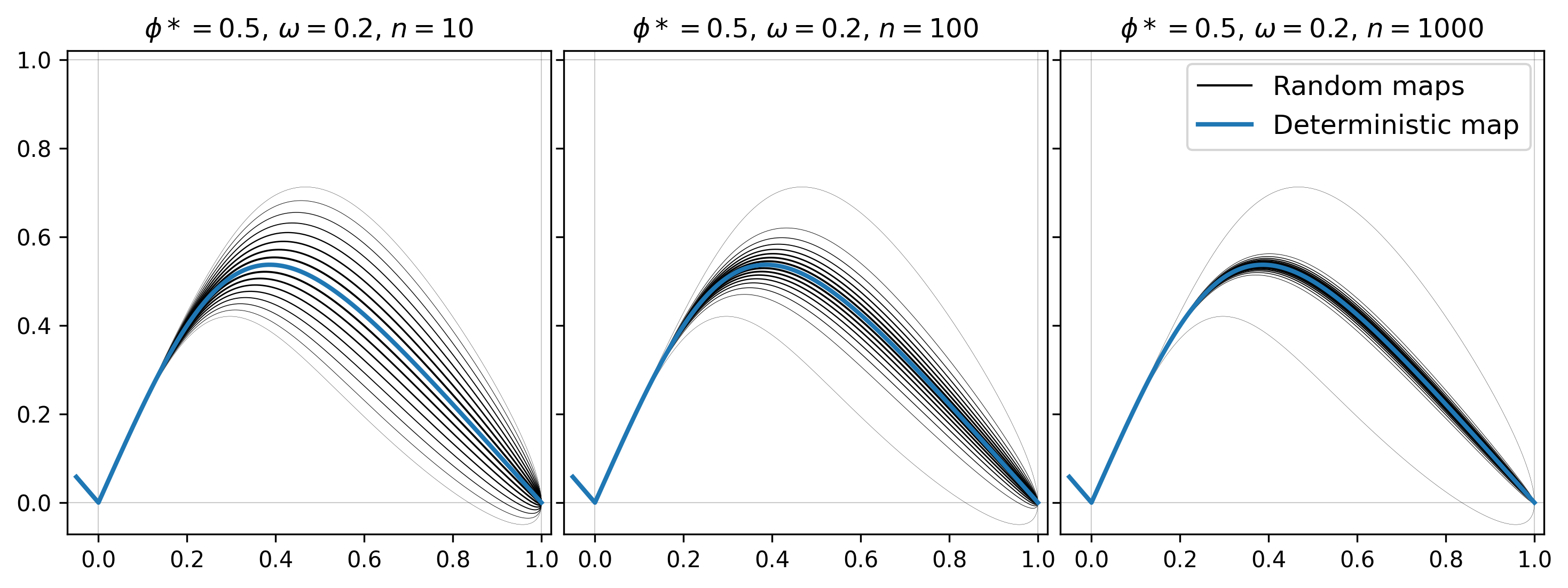}
\caption{Random maps \eqref{main_example_random_maps} for $\phi^*=0.5$, $\omega=0.2$, $\n=10, 100, 1000$, $\eta=k/16$.}
\label{fig:random_maps}
\end{figure}

As follows from \eqref{Proc2}, the random maps for the main example have the form
\begin{equation} \label{main_example_random_maps}
T_\eta(x) = T(x) + q_{\n}(\eta) \sigma_{\n}(x),
\end{equation}
where $q_{\n}$ is the quantile function of the truncated normal distribution $\mathcal{N}_{b_\n}(0,1)$.
Indeed, since $q_{\n}$ maps the uniform measure on $[0,1]$ to the truncated standard Gaussian measure on $[-b_{\n},b_{\n}]$, we have
\begin{align*}
P_x(A)
=& \mathbb{P}\{T(x) + \sigma_{\n}(x)Z_{t}\in A\}\\
=& \Leb\{\eta : T(x) + \sigma_{\n}(x)q_{\n}(\eta)\in A\}\\
=& \Leb\{\eta : T_\eta(x)\in A\}.
\end{align*}
We can equivalently rewrite \eqref{main_example_random_maps} as
\begin{equation} \label{main_example_random_maps_alt}
T_\eta(x) = T(x) + \tilde{q}_{\n}(\eta) \sigma(x),
\end{equation}
where $\tilde{q}_{\n}$ is the quantile function of $\mathcal{N}_{b_1}(0,\frac{1}{\n})$.
Notice that for different $\n$, the set $\{T_\eta\}_{\eta\in[0,1]}$ consists of the same maps, however the ones close to $T = T_{\frac{1}{2}}$ are given bigger weights for large $\n$.
More precisely, for every $\delta > 0$, we have
\begin{equation}\label{cat}
\sup_{\eta\in[\delta, 1-\delta]} \sup_{x\in I}|T_\eta(x) - T(x)|
\le \sigma_{\max} \sup_{\eta\in[\delta, 1-\delta]} |\tilde{q}_{\n}(\eta)|
\xrightarrow[\n\to\infty]{} 0.
\end{equation}

\section{Stationary measure} \label{StationaryMeasureSection}
We are now going to establish the existence of a unique stationary measure for the chain. 
This will be accomplished in the following steps: 
we first prove the Lasota-Yorke inequality \eqref{TT}; 
as a consequence we will get a finite number of ergodic absolutely continuous stationary measures whose supports are mutually disjoint up to sets of zero Lebesgue measure.
The uniqueness will be proved by showing that all the previous components share a measurable set of positive Lebesgue measure. 
Throughout this section $\n$ is fixed and we will omit it from notations.

\subsection{Existence}

We will first show that there are \emph{finitely many} ergodic stationary densities of bounded variation.
The following lemma will be useful in the sequel.

\begin{lemma} \label{Lh_TV_estimate_lemma}
For any $\rho\in BV$, if $C:=\esssup_{x\in\supp\rho} |p(x, \cdot)|_{TV} < \infty$, then
$$
|\cL\rho|_{TV} \le C \|\rho\|_1 \quad \text{and} \quad \|\cL\rho\|_{BV} \le (C+1) \|\rho\|_1.
$$
\end{lemma}

\begin{proof}
For the first inequality we have
\begin{align*}
|\cL \rho|_{TV}
&= \sup \sum_{i} \left|
\int p(x, y_{i+1}) \rho(x) dx - \int p(x, y_{i}) \rho(x) dx
\right| \\
&\leq \sup \sum_{i}
\int \left| p(x, y_{i+1}) \rho(x) - p(x, y_{i}) \rho(x) \right| dx \\
&\leq \int |p(x, \cdot)|_{TV} \rho(x) dx
\leq C\|\rho\|_1.
\end{align*}
The second inequality follows from the first one, since the Markov operator is an isometry, i.e.\ $\|\cL \rho\|_1 = \|\rho\|_1$ for all $\rho\in L^1$.
\end{proof}

We say that a stochastic kernel $p(x,y)$ has \emph{uniformly bounded variations} if $|p(x,\cdot)|_{TV} \in L^\infty$, i.e.\ there is $C>0$ such that $|p(x,\cdot)|_{TV} \leq C$ for almost every $x\in I$.

\begin{proposition} \label{ubv_fin_many_prop}
If the kernel $p$ has uniformly bounded variations,
then the operator $\cL$ is quasi-compact and there exist finitely many ergodic stationary measures with densities in $BV$
and, moreover, their supports are mutually disjoint up to sets of zero Lebesgue measure.
\end{proposition}

\begin{proof}
By Lemma~\ref{Lh_TV_estimate_lemma}, for every $n$,
$$
\|\cL^n \rho\|_{BV} 
= |\cL^n \rho|_{TV} + \|\cL^n \rho\|_1
\leq C\|\cL^{n-1} \rho\|_1 + \|\rho\|_1
= (C+1)\|\rho\|_1.
$$
In particular,
\begin{equation}\label{TT}
\|\cL \rho\|_{BV}
\le (C+1)\|\rho\|_1
\le  \eta\|\rho\|_{BV}+(C+1)\|\rho\|_1
\end{equation}
for any $\eta<1$.
This is the {\em Lasota-Yorke inequality}. 
The latter, plus the fact that  $BV$ is compactly embedded in $L^1$, implies that the peripheral spectrum of $\cL$ is discrete and therefore the chain will admit \emph{finitely many} (at least one) absolutely continuous ergodic stationary measures, with supports that are mutually disjoint up to sets of zero Lebesgue measure. 
Moreover, the essential spectral radius is strictly smaller than the  spectral radius ({\em spectral gap}).
These properties, which are consequences of the Ionescu-Tulcea-Marinescu theorem, are summarized by saying that the operator $\cL$ acting on $BV$ is {\em quasi-compact}, see, e.g., \cite{baladi2000positive,broise1996transformations,gora1997laws} for an exhaustive presentation of these results and  \cite[Section~2.3]{bahsoun2014pseudo} for a specific  application to random systems.
\end{proof}

\begin{remark}\label{cro}
\normalfont
Let us mention that whenever the operator $\cL$ is quasi-compact and the largest eigenvalue, which is $1$ in our case, is simple and therefore there is only one stationary measure with density in $BV$, then the norm of $\|\cL^k f\|_{BV}$ goes exponentially fast to zero when $k\to \infty$, for $f\in BV$ and  $\int fdx=0$ (exponential decay of correlations).
This fact will be extensively used in Section~\ref{LyapunovExponentSection}.
\end{remark}

Since the variance \eqref{noise_variance_rewritten} vanishes at 0 and 1, the kernel \eqref{main_ex_st_kernel} is in fact unbounded.
However, we can still apply Proposition~\ref{ubv_fin_many_prop} under a suitable restriction of the domain of $\cL.$ We first state a general result which allows us to confine the stationary measures.
For any $\varepsilon > 0$, define the interval
$$
\II:= [\varepsilon, 1-\Gamma/2].
$$

\begin{lemma} \label{support_lemma1}
Under the assumptions of Section~\ref{Coupling_section}, there is $\varepsilon>0$ such that any stationary measure $\mu$ has $\supp\mu\subset \{0\}\cup\II$.
If $\mu$ is continuous, $\supp\mu\subset\II$.
\end{lemma}

\begin{proof}
First, notice that any stationary measure is supported on the interval $K_\Gamma:=[-\Gamma/2, 1-\Gamma/2]$.
Indeed, by invariance,
$\mu(K_\Gamma^c) = \int P_x(K_\Gamma^c) d\mu(x) = 0$,
because $P_x(K_\Gamma^c)=0$ for all $x$.

Fix $\varepsilon>0$ such that $T(x)-s(x) > x$ for $x\in(0, \varepsilon)$.
By choosing a smaller $\varepsilon$ if needed, we may also assume that $T(x)-s(x) > \varepsilon$ for $x\in[\varepsilon, 1-\Gamma/2]$.
Then $\inf\supp P_x > \min\{x, \varepsilon\}$ for every $x\in K_\Gamma\setminus\{0\}$.
This means that for any realization $(x_t)$ of the process, either all $x_t=0$ (clearly, $0$ is a fixed point, since $P_0 = \delta_0$) or \emph{all but finitely many} $x_t>\varepsilon$.

On the other hand, if $\mu([-\Gamma, \varepsilon)\setminus\{0\})>0$, then by the Poincar\'{e} recurrence theorem, applied to the shift on $(\Omega, \mathbb{P}_\mu)$, $\mathbb{P}_\mu$-almost surely there would exist a realization $(x_t)$ with \emph{infinitely many} $0\neq x_t<\varepsilon$, which is not possible, as we showed above.
This finishes the proof.
\end{proof}

\begin{theorem} \label{cor_main_ex_bounded}
The chain defined in Section~\ref{Coupling_section} admits finitely many ergodic stationary measures with densities in $BV$.
Moreover, there is $\varepsilon>0$ such that $\supp\mu \subset \II$ for any such measure $\mu$.
\end{theorem}

\begin{proof}
By Lemma~\ref{support_lemma1}, the density of any absolutely continuous stationary measure belongs to the subspace $Y:=\{h\in L^1\mid \supp h\subset\II\}$.
From the first part of the proof of Lemma~\ref{support_lemma1} it also follows that $Y$ is $\cL$-invariant.
Moreover, the kernel \eqref{pn} has uniformly bounded variations when restricted to $\II\times \II$.
Indeed, $|p_\n(x,\cdot)|_{TV} = |g_{x,\n}|_{TV}$ and the latter is bounded on $\II$ by \eqref{g_TV_bound}.
We can therefore apply Proposition~\ref{ubv_fin_many_prop}.
\end{proof}

\begin{remark}
\normalfont
It is worth noticing that the preceding result is completely independent of the structure of the unimodal map $T$.
In this respect we could consider maps admitting attracting periodic points or Cantor sets of measure zero, but still producing smooth stationary measures when perturbed with our additive noise.
\end{remark}

\subsection{Uniqueness} \label{UniquenessSection}

We begin with the following simple lemma that links the topological dynamics of $T$ with the structure of any stationary measure.

\begin{lemma} \label{key_top_lemma}
For any stationary measure $\mu$ and any open set $U$, if $\mu(U) = 0$, then also $\mu(T^{-k}U)=0$ for all $k>0$.
\end{lemma}

\begin{proof}
It is enough to show that $\mu(T^{-1}U)=0$, the result then follows by induction.
By invariance we have $0 = \mu(U) = \int P_x(U) d\mu(x) \ge \int_{T^{-1}U} P_x(U) d\mu(x)$.
But for every $x\in T^{-1}U$, $T(x) \in U \cap \supp P_x$, and hence $P_x(U) > 0$.
Therefore the latter integral can only be zero if $\mu(T^{-1}U)=0$.
\end{proof}

With the help of the following lemma we will show that the support of any stationary measure contains the support of the $T$-invariant measure (atomic in the periodic case).
Recall that $x\in\supp\mu$ iff $\mu(U)>0$ for any open $U\ni x$.

\begin{lemma} \label{general_support_lemma}
Let $A\subset I$ be such that (1) $T$ is topologically transitive on $A$ and 
(2) $\bigcup_{k=0}^\infty T^{-k}U = I$ for any open set $U\supset A$.
Then $A\subset\supp\mu$ for any stationary measure $\mu$.
\end{lemma}

\begin{proof}
Given an open set $U$ with $U\cap A \neq \emptyset$,
by transitivity $A\subset \bigcup_{k=0}^\infty T^{-k}U$, and therefore $\bigcup_{k=0}^\infty T^{-k}U = I$.
By Lemma~\ref{key_top_lemma}, $\mu(U)>0$.
\end{proof}

Now we are ready to state the main result of this section.

\begin{theorem}
If $T$ is either periodic or chaotic, then the chain defined in Section~\ref{Coupling_section} admits a unique stationary measure $\mu$ with $BV$ density.
Moreover, $\supp\mu$ contains a neighbourhood of the periodic cycle if $T$ is periodic, or the interval $[T(\Delta), \Delta]$ if $T$ is chaotic.
\end{theorem}

\begin{proof}
If $T$ is periodic, a globally attracting cycle $\mathcal{O}$ satisfies the assumptions of Lemma~\ref{general_support_lemma}, therefore $\mathcal{O}\subset\supp\mu$ for any stationary measure $\mu$.
Let us show that $\supp\mu$ contains an open neighbourhood of $\mathcal{O}$.
Recall that $p(x,y)>0$ if and only if $y\in(T(x)-s(x),T(x)+s(x))$, and $T$ and $s$ are continuous.
Given $x_0\in\mathcal{O}$, let $x_{1}\in\mathcal{O}$ be such that $x_0=T(x_{1})$.
Since $(x_{1},x_0)\in\{(x,y)\mid p(x,y)>0\}$ and the latter set is open, we can find open sets $U\ni x_{1}$ and $V\ni x_0$ such that $p(x,y)>0$ for all $x\in U$, $y\in V$.
Also $\mu(U)>0$, because $x_{1}\in\supp\mu$.
Denoting $h$ the density of $\mu$, by invariance we get $h(y) \ge \int_U p(x,y)d\mu(x) >0$ for all $y\in V$, i.e.\ $V\subset\supp\mu$.

If $T$ is chaotic, then from \ref{At} we know that the set $I_{\Delta}=[T(\Delta), \Delta]$ is invariant for $T$ and that $T$ is topologically transitive when restricted to $I_{\Delta}$.
It also follows from \ref{itm:C3} that $\bigcup_{k=0}^\infty T^{-k} I_\Delta = I$,
so again we can apply Lemma~\ref{general_support_lemma}.

In both cases, we conclude that $\Leb(\supp\mu_1\cap\supp\mu_2)>0$ for any stationary measures $\mu_1$, $\mu_2$, and therefore by Proposition~\ref{ubv_fin_many_prop} they must coincide.
Hence the stationary measure is unique.
\end{proof}

\section{Stochastic stability} \label{StochasticStabilitySection}

Once we consider random perturbations of a deterministic dynamics, an important question is to investigate the stochastic stability of the system, which means to determine if a sequence of stationary measure will converge, in a sense to precise, to the invariant measure of the unperturbed map.
In our case the sequence of probability measures is  given by $\mu_\n:=h_\n dx$.
These measures belong to the set of Borel probability measures on the unit interval, which is a compact metric space with the weak-* topology\footnote{The \emph{weak-* topology} is given by the family of seminorms $\|\rho\|_\varphi = \int \varphi d\rho$, $\rho\in\cM$, $\varphi\in C^0$.}.
There will be therefore at least one subsequence $(\mu_{\n_k})_{k\ge 1}$ converging to a probability measure $\mu_{\infty}$ on $I$.
Our objective is to prove that: 
(i) $\mu_{\infty}$ is invariant, 
(ii) it is the same for any convergent subsequence, if more than one, and 
(iii) it coincides with $\mu$.
Whenever that happens we will say that our random system is  {\em weakly stochastic stable}. 
This result could be strengthened  by showing that $\|h_\n-h\|_1\to 0$, which is called the {\em strong stochastic stability}; we are not able at the moment to get this result. 
Instead we now give a {\em sufficient condition} to get the weak stochastic stability:
\begin{enumerate}
\item[\optionalitemlabel{(A$_q$)}{Aq}]
\textit{There exist $q>1$ and $C_q>0$ such that for all $\n\ge 1$ we have $\|h_\n\|_q\le C_q$.}
\end{enumerate}

We will see in the next section that with the preceding assumption we can prove the convergence of the Lyapunov exponent (Theorem~\ref{Lexp_conv_n}) and then verify it numerically, which is an indirect indication of the validity of \ref{Aq}.

\begin{lemma}\label{weak_star_delta_lemma}
For every $x\in I$, $P_x^{(\n)}$ converges to $\delta_{Tx}$ in the weak-* topology as $\n\to\infty$, i.e.\ $\int \varphi dP_x^{(\n)} \to \varphi(Tx)$ for all $\varphi\in C^0(I)$.
\end{lemma}

\begin{proof}
For arbitrary $\varepsilon>0$ we can split
$$
\int \varphi dP_x^{(\n)} = 
\int_{B_\varepsilon(Tx)} \varphi dP_x^{(\n)} + 
\int_{B_\varepsilon(Tx)^c} \varphi dP_x^{(\n)}.
$$
By Chebyshev's inequality, $P_x^{(\n)}(B_\varepsilon(Tx)^c)\le {\mathrm{Var}^2 P_x^{(\n)}}/{\varepsilon^2}\to 0$ as $\n\to\infty$, while $\varphi$ is bounded, so the second integral can be made arbitrarily small for $\n$ large.
Consequently, $P_x^{(\n)}(B_\varepsilon(Tx))\to 1$, and since $\varphi$ is continuous, the first integral can be made arbitrarily close to $\varphi(Tx)$.
\end{proof}

The proof of \ref{LK1} below follows a suggestion in \cite[Theorem~D]{alves2003random}.

\begin{proposition}\label{LK}
Let $\mu$ be a weak-* limit measure of a sequence $\mu_{\n_k}=h_{\n_k}dx$.
If $h_{\n_k}$ satisfy \ref{Aq}, then 
\begin{enumerate}[label=(\roman*)]
    \item\label{LK1} $\mu$ is absolutely continuous with density in $L^q$;
    \item\label{LK2} $\mu$ is invariant under $T$.
\end{enumerate}
\end{proposition}

\begin{proof}
\ref{LK1}
Let $\varphi\in C^0(I)$.
By H\"{o}lder's inequality, with $p=\frac{q}{q-1}$,
$$
\Bigl|\int \varphi d\mu \Bigr|
= \Bigl|\lim_{k\to \infty}\int \varphi h_{\n_k}dx \Bigr|
\le \lim_{k\to \infty}\|h_{\n_k}\|_q \|\varphi\|_p
\le C_q\|\varphi\|_p.
$$
Therefore the map $L^p\ni\varphi \mapsto \int\varphi d\mu\in\mathbb{R}$ is continuous, since $C^0$ is dense in $L^p$, and therefore such a functional will be in $L^q$, namely $\mu=hdx$, $h\in L^q$, $\|h\|_q\le C_q$. 

\ref{LK2}
It suffices to prove that any test function $\varphi\in C^0(I)$ satisfies 
$$
\int \varphi h_{\n_k} dy-\int \varphi \circ T h_{\n_k}dy\to 0.
$$
Since $h_{\n_k}=\cL h_{\n_k}=\int p_{\n_k}(x,\cdot)h_{\n_k}(x)dx$, by changing the order of integration in the first integral and subtracting the second, we get
\begin{equation}\label{gh}
\int h_{\n_k}(x) \left[\int p_{\n_k}(x,y) \varphi(y)dy - \varphi(T(x))\right]dx.
\end{equation}
Since the function $\psi_{\n_k}(x):=\int p_{\n_k}(x,y) \varphi(y)dy - \varphi(T(x))$ is uniformly bounded,
$$
\eqref{gh} \le \|h_{\n_k}\|_q \|\psi_{\n_k}\|_p\le C_q \|\psi_{\n_k}\|_p.
$$
By Lemma~\ref{weak_star_delta_lemma}, $\psi_{\n_k}(x)\to 0$ for every $x\in I$, and therefore by dominated convergence $\|\psi_{\n_k}\|_p \to 0$.
\end{proof}

\begin{theorem} \label{WSS}
Under Assumption~\ref{Aq},
the chain defined in Section~\ref{Coupling_section} is weakly stochastic stable, i.e.\ the stationary probabilities converge to the unique $T$-invariant probability in the weak-* topology as $\n\to\infty$.
\end{theorem}

\begin{proof}
Since $T$ admits a unique invariant measure, $\mu$ must be the same for all convergent subsequences in Proposition~\ref{LK}, and therefore the entire sequence $\mu_\n$ converges to $\mu$.
\end{proof}

It follows from Proposition~\ref{LK} that Assumption~\ref{Aq} cannot be satisfied in the periodic case, since the limiting $T$-invariant measure is singular and supported on the periodic orbit, so Theorem~\ref{WSS} only covers the chaotic case.
We will now give a proof in the periodic case under the following assumption:
\begin{enumerate}
\item[\optionalitemlabel{(A$_s$)}{As}]
\textit{For all $\n$ sufficiently large and all $x\in\supp\mu_\n$ we have $|T'(x)|\le\tau<1$.}
\end{enumerate}

\begin{proposition} \label{WSS_fixed_point}
If $T$ is periodic and satisfies \ref{As}, then
the chain defined in Section~\ref{MainExampleSection} is weakly stochastic stable.
\end{proposition}

\begin{proof}
Let us first consider the case when $T$ has a globally attracting fixed point $x_0$.
We need to show that for any test function $\varphi\in C^0(I)$ we have $\int \varphi(x) h_{\n}(x) dx\to \varphi(x_0)$ as $\n\to\infty$.
Since $h_\n$ is a fixed point of the random transfer operator \eqref{dis} and this operator is the dual of the random Koopman operator $\varphi\mapsto\int\varphi\circ T_\eta d\eta$ (see, e.g., \cite[Section~2]{aimino2015annealed} for details), the previous weak limit leads to prove that the following quantity
\begin{equation}\label{wss_fp_int}
\int_I \int_{[0,1]^k} \bigl(\varphi(T_{\eta_k} \circ\cdots\circ T_{\eta_1}(x)) - \varphi(x_0)\bigr) h_{\n}(x)\,d\bar{\eta}\,dx
\end{equation}
goes to 0 as $\n\to\infty$, where $k$ is an arbitrary fixed number and $\bar{\eta} = (\eta_1,\ldots,\eta_k)$.

Given $\varepsilon>0$, let $\zeta>0$ be such that $|\varphi(x) - \varphi(x_0)|<\varepsilon$ when $|x-x_0|<\frac{2\zeta}{1-\tau}$.
Fix $k$ such that, for all $\n$,
\begin{equation}\label{wss_fp_1}
\sup_{x\in\supp\mu_\n} |T^k(x) - x_0| < \zeta.
\end{equation}
Next, fix $\delta > 0$ such that
\begin{equation}\label{wss_fp_2}
2\|\varphi\|_\infty (1 - (1 - 2\delta)^k) < \varepsilon.
\end{equation}
Finally, by \eqref{cat}, for all $\n$ sufficiently large, we have
\begin{equation}\label{wss_fp_3}
\sup_{\eta\in[\delta, 1-\delta]} \sup_{x\in I} |T_\eta(x) - T(x)| < \zeta.
\end{equation}
We now split the integral \eqref{wss_fp_int} in the $\bar{\eta}$ variable over the region $E:=[\delta, 1-\delta]^k$ and its complement.
On $E^c$ the absolute value of \eqref{wss_fp_int} is bounded by \eqref{wss_fp_2}.
Notice that the integral over $x$ takes place on the support of $\mu_{\n}$, where \ref{As} holds.
Moreover, since the map $[0,1]\ni\eta\mapsto T_\eta\in C^2(I)$ is continuous, each $T_\eta$ maps $\supp\mu_\n$ to itself; see \cite{araujo2000attractors}. 
Therefore, for $\bar{\eta} \in E$ and $x\in\supp\mu_\n$, by \eqref{wss_fp_3} and \ref{As}, we have
$$
|T_{\eta_2}\circ T_{\eta_1}(x) - T^2(x)| \le
|T_{\eta_2} (T_{\eta_1}(x)) - T(T_{\eta_1}(x))| + |T(T_{\eta_1}(x)) - T^2(x)|
< \zeta + \tau\zeta.
$$
By induction we easily get
$
|T_{\eta_k}\circ\cdots\circ T_{\eta_1}(x) - T^k(x)| < \frac{\zeta}{1-\tau}
$
and therefore, in view of \eqref{wss_fp_1},
$
|T_{\eta_k}\circ\cdots\circ T_{\eta_1}(x) - x_0| < \frac{2\zeta}{1-\tau}
$
for all $\bar{\eta} \in E$ and $x\in\supp\mu_\n$.
Then, by the choice of $\zeta$, $|\varphi(T_{\eta_k}\circ\cdots\circ T_{\eta_1}(x)) - \varphi(x_0)| < \varepsilon$ and the absolute value of \eqref{wss_fp_int} over $E$ is therefore bounded by $\varepsilon$.

It is straightforward to modify the above proof for the case when $T$ has a globally attracting periodic orbit of length $m>1$.
One needs to replace $T$ with $T^m$, the latter will have $m$ attracting fixed points.
The corresponding random maps of the form $T_{\bar\eta} = T_{\eta_1}\circ\cdots\circ T_{\eta_m}$ will be parametrized by $\bar{\eta} = (\eta_1,\cdots,\eta_m)\in [0,1]^m$ endowed with the Lebesgue measure.
We leave the details to the reader.
\end{proof}

\begin{remark}
\normalfont
It follows from the proof that Proposition~\ref{WSS_fixed_point} remains valid for the general class of Markov chains defined in Section~\ref{Coupling_section} whenever \eqref{cat} holds, which is in turn the case when $T$ and $g_{x,\n}$ are sufficiently smooth.
\end{remark}

\begin{remark}
\normalfont
We conjecture that, if $T$ is periodic with the attracting periodic orbit $\mathcal{O}$, then $h_\n\to 0$ uniformly on compact sets $K\subset I\setminus\mathcal{O}$ as $\n\to\infty$.
This property, that we checked numerically, straightens the previous result.
In particular,
\begin{enumerate}
\item[\optionalitemlabel{(A$_c$)}{Ac}]
\textit{If $T$ is periodic and the critical point $c$ does not belong to the attracting periodic orbit, then $h_\n\to 0$ uniformly in a neighbourhood of $c$ as $\n\to\infty$.}
\end{enumerate}
\end{remark}

\section{Lyapunov exponent} \label{LyapunovExponentSection}

\subsection{Average Lyapunov exponent}\label{ssdd}
We are interested in the existence of the \emph{Lyapunov exponent} for the slow component, which in our case is defined $\mathbb{P}_\mu$-almost surely as the limit
\begin{equation} \label{Lexp}
\Lambda = \lim_{n\to\infty} \frac{1}{n} \sum_{t=0}^{n-1} \log|T'(X_t)|
\end{equation}
along the chain $(X_t)_{t\geq 0}$. 
We now motivate such a choice. 
It is twofold: first of all we want to reproduce the Lyapunov exponent of the unperturbed map $T$ in the limit of zero noise, which we will get in Theorem~\ref{Lexp_conv_n}; successively we want an indicator which kept memory of the underlying {\em slow dynamics} played by the map $T$.
We  will, in particular, show that such an exponent is negative for periodic $T$, even in presence of mixing stationary measure.

We now return to \eqref{Lexp}; if the chain admits a unique stationary probability $\mu$, then, by the ergodic theorem for Markov chains, the above limit equals
\begin{equation} \label{Lexp_int}
\int \log|T'| d\mu,
\end{equation}
assuming $\log|T'|\in L^1(\mu)$.

\begin{remark}
\normalfont
The Lyapunov exponent \eqref{Lexp_int} was called the {\em average Lyapunov exponent} in \cite{galatolo2020existence, nisoli2020sufficient}, and it was associated to the phenomenon of  {\em noise induced order}, which happens when the perturbed systems admit a unique stationary measure depending on some parameter, say $\theta$, and the Lyapunov exponent depends continuously on $\theta$ and exhibits a transition from  positive to negative values, see also \cite{matsumoto1983noise} for an experimental evidence of this fact. 
We will partially prove  this phenomenon below by combining Corollary \ref{cor_pos_neg_Lexp} and Theorem~\ref{Lexp_cont_thm}, and show it numerically in Section~\ref{NumericalLyapunovSection}. 
\end{remark}

A unimodal map $T$ is said to have a \emph{critical point of order $l$} if there is a constant $D$ such that $D^{-1}|x-c|^{l-1}\le |T'(x)|\le D|x-c|^{l-1}$.
In this case it was proved in \cite{nowicki1991invariant} that the invariant density for $T$ is in $L^q$, with $q<\frac{l}{l-1}$. 
We assume in \ref{itm:A3} that $T$ has a critical point of order 2.
It is easy to check that \eqref{unimodal_map_rewritten} satisfies this assumption.
Consequently, $\log|T'|$ is in $L^p$ for any $p\ge 1$.

\begin{theorem}
If $T$ is periodic or chaotic, the limit \eqref{Lexp} exists almost surely. 
\end{theorem}

\begin{proof}
The integral \eqref{Lexp_int} is finite, because $\log|T'|$ is in $L^1$ and the unique stationary measure $\mu$ has bounded density, as we proved in Section~\ref{StationaryMeasureSection}. 
\end{proof}

Once we know that the Lyapunov exponent exists almost surely, it is natural to ask how it depends on the model parameters, for instance, the length $\n$ of the fast component series.
We have the following

\begin{theorem}\label{Lexp_conv_n}
Suppose one of the following is satisfied:
(a) $T$ verifies \ref{Aq};
(b) $T$ is periodic and verifies \ref{As} and \ref{Ac}.
Then the Lyapunov exponent \eqref{Lexp} converges to the Lyapunov exponent of the deterministic map $T$ as $\n\to\infty$.
\end{theorem}

\begin{proof}
(a)
Denote with $\mu_\n=h_\n dx$ the unique stationary measure associated to $\n$ and with $\mu=hdx$ the unique invariant measure for $T$.
We need to show that
\begin{equation}\label{Lexp_n_conv}
\int \log |T'|d\mu_\n \xrightarrow[\n\to\infty]{} \int \log |T'|d\mu. 
\end{equation}
Let $q>1$ be such that $h_\n, h \in L^q$ and set $p:=\frac{q}{q-1}$.
Since $\log|T'|\in L^p$, for any $\varepsilon>0$ there is $\varphi_{\varepsilon}\in C^0$ such that $\|\log|T'|-\varphi_{\varepsilon}\|_p<\varepsilon$.
Write 
$$
\int \log |T'|d\mu_\n=\int (\log |T'|-\varphi_{\varepsilon})  h_\n dx + \int \varphi_{\varepsilon}d\mu_\n
$$
and
$$
\int \log |T'|d\mu=\int (\log |T'|-\varphi_{\varepsilon})  hdx + \int \varphi_{\varepsilon}d\mu.
$$
Since $\log |T'|-\varphi_{\varepsilon}\in L^p$ and $h_\n,h\in L^q$, we have
$$
\int |\log |T'|-\varphi_{\varepsilon}| h_\n dx \le 
\|\log|T'|-\varphi_{\varepsilon}\|_p\|h_\n\|_q\le \varepsilon C_q,
$$
and the same inequality holds for the integral with respect to $\mu$.
Finally, from Theorem~\ref{WSS} we know that $\mu_\n \xrightarrow[]{w^*} \mu$, hence 
$\int \varphi_{\varepsilon} d\mu_\n \to \int \varphi_{\varepsilon} d\mu$ as $\n\to\infty$.

(b)
We know from Proposition~\ref{WSS_fixed_point} that $\int \varphi d\mu_\n \to \frac{1}{|\mathcal{O}|}\sum_{x\in\mathcal{O}}\varphi(x)$ for all $\varphi\in C^0(I)$ as $\n\to\infty$, since the $T$-invariant measure $\mu$ is atomic and supported on the attracting periodic orbit $\mathcal{O}$.
Let us first consider the case when the critical point $c$ belongs to $\mathcal{O}$;
the right-hand side of \eqref{Lexp_n_conv} is then $-\infty$.
Denoting $f_m(x) := \max\{\log|T'(x)|, -m\}\in C^0(I)$,
for every $m$ we have
$$
\int \log|T'| d\mu_\n \le \int f_m d\mu_\n \xrightarrow[\n\to\infty]{} \frac{1}{|\mathcal{O}|}\sum_{x\in\mathcal{O}}f_m(x) \le - \frac{m}{|\mathcal{O}|} + C,
$$
because $f_m\le C:= \sup\log|T'|$ and $f_m(c) = -m$.
Therefore $\int \log|T'| d\mu_\n \to -\infty$ as $\n\to\infty$.

If $c\notin\mathcal{O}$, we can fix a neighbourhood $U\ni c$ given by \ref{Ac} and split 
$$
\int \log|T'| d\mu_\n = \int_U \log|T'| d\mu_\n + \int_{U^c} \log|T'| d\mu_\n.
$$
The first term is bounded by $\|\log|T'|\|_1 \sup_U h_\n$ and vanishes as $\n\to\infty$ by \ref{Ac},
while the second one converges to $\int \log|T'| d\mu$ by Proposition~\ref{WSS_fixed_point} (approximate $\log|T'|\1_{U^c}$ with a suitable continuous function).
\end{proof}

\begin{corollary}\label{cor_pos_neg_Lexp}
Under the assumptions of Theorem~\ref{Lexp_conv_n} and for $\n$ large enough, $\Lambda$ is positive if $T$ is chaotic, and negative if $T$ is periodic.
\end{corollary}

In a few cases the negativity of the Lyapunov exponent can be shown relatively easily.
For instance, denote $\{ x\in I \mid |T'(x)|\le 1 \} = [m,M]$ and $\bar\Delta := \sup_{x\in I} T(x)+s(x)$.
If $T(x)-s(x) > \min\{x, m\}$ for all $x\in(0,\bar\Delta]$, then, arguing as in the proof of Lemma~\ref{support_lemma1}, one can show that any continuous stationary measure $\mu$ has $\supp\mu\subset [m, \bar\Delta]$.
So if, moreover, $\bar\Delta \le M$, then $\Lambda < 0$.
Following the classification given in Section~\ref{UnimodalMapsSection}, let us consider the case $T(c) < c$, where the map $T$ exhibits a globally attracting fixed point.
In this case, the conditions above will be satisfied if $s(x)$ is small enough,
in other words, the stationary measure will be supported in a neighbourhood of the fixed point, where $|T'(x)| \le 1$.
For other cases, we provide some numerical examples in Section~\ref{NumericalLyapunovSection}.

\subsection{Continuity of the Lyapunov exponent} \label{LyapExpContinuitySection}

Denote by $\Theta := \{\theta = (\phi^*, \omega, \n)\in(0,1)^2\times(0,\infty) \mid \max T_\theta < 1\}$ the (extended) parameter space.
In order to prove the continuity of the Lyapunov exponent, we will assume that $T_\theta(x) \in C^3(\Theta\times [0,1])$ and $p_\theta(x,y) \in C^2(\Theta\times (0,1)^2)$.
It is straightforward that our main example defined in Section~\ref{MainExampleSection} satisfies this assumption.
Let $\tilde{\Theta}\subset\Theta$ be the set of parameters $\theta$ for which there is a unique stationary measure $\mu_\theta$ with a density $h_\theta\in BV$;
we proved in Section~\ref{StationaryMeasureSection} that this is the case if $T_\theta$ is periodic or chaotic, but our numerical investigations confirm that in fact $\Leb(\Theta\setminus\tilde{\Theta})=0$.

\begin{theorem} \label{Lexp_cont_thm}
The mapping $\tilde{\Theta}\ni\theta\mapsto \Lambda_{\theta}\in\mathbb{R}$ is continuous.
\end{theorem}

\begin{proof}
Fix an exhaustion of $\Theta$ by nested compact sets $\Theta_\iota$ and set $\tilde{\Theta}_\iota := \Theta_\iota\cap\tilde{\Theta}$.
It is enough to prove that the mapping $\tilde{\Theta}_\iota \ni \theta \mapsto \Lambda_\theta \in\mathbb{R}$ is continuous on each $\tilde{\Theta}_\iota$, and from now on we fix one of them.
As we showed in Lemma~\ref{support_lemma1}, for each $\theta\in\Theta$, $\supp\mu_\theta\subset I_{\varepsilon_\theta} = [\varepsilon_\theta, 1-\varepsilon_\theta]$, and since $\varepsilon_\theta$ can be shown to depend continuously on $\theta$, we can find a single $\varepsilon>0$ that works for all $\theta\in\Theta_\iota$.

Given $\theta,\theta'\in \tilde{\Theta}_\iota$ we can write
$$
|\Lambda_{\theta}-\Lambda_{\theta'}| 
= \Bigl|\int \log |T_{\theta}'|h_{\theta}dx-\int \log |T_{\theta'}'|h_{\theta'}dx\Bigr| \le
$$
\begin{equation}\label{cont_1}
\int \bigl|\log |T_{\theta}'|\bigr| |h_{\theta}-h_{\theta'}|dx 
+ \int \bigl|\log |T'_{\theta}|-\log |T_{\theta'}'|\bigr| h_{\theta'} dx.
\end{equation}
To bound the second term in \eqref{cont_1}, first notice that, by Lemma~\ref{Lh_TV_estimate_lemma},
$$
\|h_{\theta'}\|_\infty \le \|h_{\theta'}\|_{BV} = \|\cL_{\theta'} h_{\theta'}\|_{BV} \le C \|h_{\theta'}\|_1 = C,
$$
where 
$$
C = 1+\sup_{\theta\in\Theta_\iota} \sup_{x\in I_\varepsilon} |p_\theta(x,\cdot)|_{TV} \le 1+\sup_{\theta\in\Theta_\iota} \sup_{x,y\in I_\varepsilon} \bigl|\frac{\partial p_\theta}{\partial y}(x,y)\bigr| < \infty
$$
is finite, because $\frac{\partial p_\theta}{\partial y}(x,y)$ is continuous and $\Theta_\iota\times I_\varepsilon^2$ is compact.
The second term is thus bounded by $C \|\log |T'_{\theta}|-\log |T_{\theta'}'|\|_1$
and, by Lemma~\ref{logTv_lemma} below, goes to 0 as $\theta'\to\theta$.

We now estimate the first term in \eqref{cont_1}.
Since $\log|T'_\theta| \in L^1$, it is enough to bound $\|h_{\theta}-h_{\theta'}\|_\infty$ which is again dominated by $\|h_{\theta}-h_{\theta'}\|_{BV}$.
By invariance,
$$
\|h_{\theta}-h_{\theta'}\|_{BV}
= \|\cL^k_{\theta} h_{\theta}-\cL^k_{\theta'}h_{\theta'}\|_{BV}
\le \|\cL^k_{\theta} (h_{\theta} - h_{\theta'})\|_{BV}
+ \|(\cL^k_{\theta} - \cL^k_{\theta'}) h_{\theta'}\|_{BV}.
$$
As we said in Remark~\ref{cro}, the Markov operator $\cL_\theta$ enjoys the exponential bound
$$
\|\cL_\theta^k f\|_{BV}\le C_\theta \zeta_\theta^k \|f\|_{BV}
$$
for all $k>0$ and $f\in BV$ supported on $I_\varepsilon$ with $\int f dx=0$, where the constants $C_\theta > 0$, $0<\zeta_\theta<1$ depend on the parameter $\theta$.
Since $h_\theta - h_{\theta'}$ has zero mean, we therefore have
$$
\|\cL^k_{\theta} (h_{\theta} - h_{\theta'})\|_{BV}\le C_{\theta} \zeta_{\theta}^k \|h_{\theta}-h_{\theta'}\|_{BV}.
$$
Then 
$$
(1- C_{\theta} \zeta_{\theta}^{k}) \|h_{\theta}-h_{\theta'}\|_{BV}
\le \|(\cL^k_{\theta} - \cL^k_{\theta'}) h_{\theta'}\|_{BV},
$$
and for $k$ sufficiently large, $C_{\theta} \zeta_{\theta}^{k}<1$.
By a standard trick, expanding a telescopic sum $\cL^k_{\theta} - \cL^k_{\theta'} = \sum_{j=1}^k \cL_{\theta}^{k-j}(\cL_{\theta}-\cL_{\theta'})\cL_{\theta'}^{j-1}$, we get
\begin{multline*}
\|(\cL^k_{\theta} - \cL^k_{\theta'}) h_{\theta'}\|_{BV}
\le \sum_{j=1}^k \|\cL_{\theta}^{k-j}(\cL_{\theta}-\cL_{\theta'})h_{\theta'}\|_{BV} \\
\le \sum_{j=1}^k C_{\theta} \zeta_{\theta}^{k-j}\|(\cL_{\theta}-\cL_{\theta'})h_{\theta'}\|_{BV} 
\le C_{\theta}\frac{1}{1-\zeta_{\theta}}\|(\cL_{\theta}-\cL_{\theta'})h_{\theta'}\|_{BV}.
\end{multline*}
Combining the above inequalities we come to the following estimate:
$$
\int \bigl|\log |T_{\theta}'|\bigr| |h_{\theta}-h_{\theta'}|dx \le M_\theta \|(\cL_{\theta}-\cL_{\theta'})h_{\theta'}\|_{BV},
$$
where $M_\theta = \frac{C_\theta \|\log|T'_\theta|\|_1}{(1 - \zeta_\theta)(1- C_\theta \zeta_\theta^{k})}$.
It therefore remains to bound $\|(\cL_{\theta}-\cL_{\theta'})h_{\theta'}\|_{BV}$.
Since both $h_{\theta'}$ and $(\cL_{\theta}-\cL_{\theta'})h_{\theta'}$ are supported on $I_\varepsilon$ and $\|h_{\theta'}\|_\infty \le C$, we have
$$
\|(\cL_{\theta}-\cL_{\theta'})h_{\theta'}\|_1 
= \int_{I_\varepsilon} \int_{I_\varepsilon} |p_{\theta}(x,y) - p_{\theta'}(x,y)| h_{\theta'}(x) dx dy
\le M_1 \|\theta-\theta'\|,
$$
where $M_1 = C \sup_{\theta\in\Theta_\iota} \sup_{x,y\in I_\varepsilon} \|\nabla_\theta p_\theta(x,y)\|$ is finite because $\nabla_\theta p_\theta(x,y)$ is continuous and $\Theta_\iota\times I_\varepsilon^2$ is compact.
Similarly, arguing as in the proof of Lemma~\ref{Lh_TV_estimate_lemma}, we get
$$
|(\cL_{\theta}-\cL_{\theta'})h_{\theta'}|_{TV} 
\le \int_{I_\varepsilon} |p_{\theta}(x, \cdot) - p_{\theta'}(x, \cdot)|_{TV} h_{\theta'}(x) dx
\le M_2 \|\theta-\theta'\|,
$$
with $M_2 = C \sup_{\theta\in\Theta_\iota} \sup_{x,y\in I_\varepsilon} \|\nabla_\theta \frac{\partial p_\theta}{\partial y}(x,y)\| < \infty$.
Therefore, the first term in \eqref{cont_1} is bounded by $M_\theta(M_1+M_2)\|\theta-\theta'\|$.
This finishes the proof.
\end{proof}

\begin{remark}
\normalfont
Clearly, the above proof works if we replace $\log|T'|$ with any continuous function.
Therefore, the mapping $\tilde{\Theta}\ni\theta\mapsto \mu_{\theta}\in\cM$ is continuous with respect to the weak-* topology on $\cM$,
i.e.\ $\tilde{\Theta}\ni\theta\mapsto \int\varphi\mu_{\theta}\in\mathbb{R}$ is continuous for any $\varphi\in C^0(I)$.
Theorem~\ref{Lexp_cont_thm} is more delicate, however, because $\log|T'_\theta|$ is neither continuous nor bounded and also depends on $\theta$.
\end{remark}

\begin{lemma} \label{logTv_lemma}
$\|\log |T'_{\theta}|-\log |T_{\theta'}'|\|_1\to 0$ as $\theta'\to\theta$.
\end{lemma}

\begin{proof}
First notice that the critical point $c_\theta$ of the map $T_\theta$ depends continuously on the parameter $\theta\in\Theta$.
Indeed, since $T'_\theta$ is continuous on $\Theta\times I$, the set $\{(\theta, c_\theta)\} = (T'_\theta)^{-1}(\{0\})$ is closed, and then the map $\Theta\ni\theta\mapsto c_\theta\in I$ is continuous by the closed graph theorem.

The functions $\log|T'_\theta|$ and $\log|T'_{\theta'}|$ have logarithmic singularities at $c_\theta$ and $c_{\theta'}$ respectively.
We will show that these singularities cancel out as $c_{\theta'}$ approaches $c_\theta$.
As in the proof of Theorem~\ref{Lexp_cont_thm}, we may assume that $\theta, \theta' \in \Theta_\iota$, where $\Theta_\iota$ is compact.

Let $\alpha := \sup_{\theta\in\Theta_\iota}|T'''_\theta|$.
As a direct consequence of the mean value theorem,
$$
(|T''_\theta(c_\theta)|-2\alpha\delta) |x-c_\theta|
\le |T'_\theta(x)| \le 
(|T''_\theta(c_\theta)|+2\alpha\delta) |x-c_\theta|
$$
for all $|x-c_\theta| \le 2\delta$, and the same inequality holds if we replace $\theta$ with $\theta'$.
Set $D_\delta^\pm := |T''_\theta(c_\theta)|\pm 3\alpha\delta$; both $D_\delta^+$ and $D_\delta^-$ are positive, since $T_\theta$ has quadratic critical point ($T''_\theta(c_\theta)<0$).
If $\theta'$ is sufficiently close to $\theta$, then $|c_\theta - c_{\theta'}| < \delta/2$ and $|T''_\theta(c_\theta) - T''_\theta(c_{\theta'})|<\delta$,
and for all $|x-c_\theta| \le \delta$ we then simultaneously have
$$
D_\delta^- |x-c_\theta| \le |T'_\theta(x)| \le D_\delta^+ |x-c_\theta|,
$$
$$
D_\delta^- |x-c_{\theta'}| \le |T'_{\theta'}(x)| \le D_\delta^+ |x-c_{\theta'}|,
$$
and hence
$$
\bigl|\log |T'_{\theta}|-\log |T_{\theta'}'|\bigr|
\le \log\frac{D_\delta^+}{D_\delta^-} 
+ \Bigl|\log\frac{|x-c_{\theta}|}{|x-c_{\theta'}|}\Bigr|.
$$
Given $\varepsilon>0$ and $\theta\in\Theta_\iota$, we first fix $\delta>0$ such that $\log\frac{D_\delta^+}{D_\delta^-} < \frac{\varepsilon}{3}$ and then let $\theta'\to\theta$.
The integral of the second term is elementary and vanishes as $c_{\theta'}\to c_{\theta}$,
so $\int_{B_\delta(c_\theta)} \bigl|\log\frac{|x-c_{\theta}|}{|x-c_{\theta'}|}\bigr| dx < \frac{\varepsilon}{3}$, provided $\theta'$ and $\theta$ are sufficiently close,
and
\begin{equation} \label{log_est1}
\int_{B_\delta(c_\theta)} \bigl|\log |T'_{\theta}|-\log |T_{\theta'}'|\bigr| dx < \frac{2\varepsilon}{3}.
\end{equation}

Let us show that $\log |T'_{\theta}|-\log |T_{\theta'}'|\to 0$ uniformly on $B_\delta(c_\theta)^c$.
Denote $\beta := \inf_{x\in B_\delta(c_\theta)^c}|T'_\theta(x)| > 0$.
We have $\|T'_{\theta}-T'_{\theta'}\|_\infty \le M \|\theta-\theta'\|$, where $M = \sup_{\theta\in\Theta_\iota, x\in I}\|\nabla_\theta T'_\theta(x)\|$.
Therefore $\inf_{x\in B_\delta(c_\theta)^c}|T'_{\theta'}(x)| > \frac{\beta}{2}$ if $\|\theta-\theta'\| < \frac{\beta}{2M}$.
Consequently, $\bigl|\log |T'_{\theta}|-\log |T_{\theta'}'|\bigr| < \frac{2M}{\beta} \|\theta-\theta'\|$, and for $\|\theta-\theta'\| < \frac{\varepsilon\beta}{6M}$ we have
\begin{equation} \label{log_est2}
\int_{B_\delta(c_\theta)^c} \bigl|\log |T'_{\theta}|-\log |T_{\theta'}'|\bigr| dx < \frac{\varepsilon}{3}.
\end{equation}
Combining \eqref{log_est2} with \eqref{log_est1} we finally get $\|\log |T'_{\theta}|-\log |T_{\theta'}'|\|_1 < \varepsilon$.
\end{proof}

\subsection{Random Lyapunov exponent and random entropy}\label{RLESection}

In \eqref{Lexp} we used the derivative of the deterministic map only.
Alternatively, if we define the process using the random transformations  \eqref{rmp}, we are led to compute the Lyapunov exponent of the cocycle given by the derivative computed along the random orbit, namely we define the {\em random Lyapunov exponent} (RLE) $\bar\Lambda$ as
\begin{equation}\label{RE}
\bar\Lambda:=\lim_{n\to\infty} \frac{1}{n}\log|D(T_{\eta_n}\circ \cdots \circ T_{\eta_1})(x)|,
\end{equation}
for almost every sequence $(\eta_k)\in [0,1]^{\mathbb{N}}$ with respect to the measure $\Leb^{\otimes\mathbb{N}}$ (see Section~\ref{RandomTransformSection}), and almost every  $x\in I$ with respect to the stationary measure $\mu_{\n}$. 
By using the notation introduced in Section~\ref{RandomTransformSection}: $\bar x_k(x):=T_{\eta_k}\circ \cdots \circ T_{\eta_1}(x)$, with $\bar x_0(x)=x$,  formula \eqref{RE} is equal to
\begin{equation} \label{Lexp_rand_maps}
\bar\Lambda = \lim_{n\to\infty} \frac{1}{n} \sum_{k=1}^{n} \log|T_{\eta_k}'(\bar x_{k-1})|,
\end{equation}
again for $\Leb^{\otimes\mathbb{N}}$-a.e.\ $(\eta_k)\in [0,1]^{\mathbb{N}}$ and $\mu_{\n}$-a.e.\ $x\in I$.
By using the ergodic theorem for random transformations, see \cite{kifier1986ergodic} or \cite[Section~3.1]{araujo2005stochastic}, we have that 
$$
\bar\Lambda=\int \log|T'_{\eta}(x)|d\mu_{\n}(x) d\eta.
$$
Notice that if we compare this random exponent with $\Lambda$, we see that the difference between the two is bounded by
$$
|\Lambda-\bar\Lambda|\le \int \bigl|\log|T'_{\eta}(x)|-\log|T'(x)|\bigr|d\mu_{\n}(x)d\eta.
$$
By using expression \eqref{main_example_random_maps_alt} for $T_{\eta}$, we can bound the error term in a more explicit way as
$$
|\Lambda-\bar\Lambda|\le \int\left|\log\Bigl|1+\frac{\tilde{q}_{\n}(\eta) \sigma'(x)}{T'(x)}\Bigr|\right|h_{\n}(x) dx d\eta.
$$
Since $\sigma'$ is bounded over $I_{\varepsilon, \Gamma}$, $\log|T'|\in L^p, p\ge1$, and finally $\|h_{\n}\|_q\le C_q$ for all $\n$, we have
$$
|\Lambda-\bar\Lambda|\le C_q \int \left\|\log \Bigl|1+\frac{\tilde{q}_{\n}(\eta) \sigma'}{T'}\Bigr|\right\|_p d\eta.
$$
We expect that this error converges to zero for large $\n$, since the quantile function converge to zero almost everywhere.
Table~\ref{tab:lyap_exp_conv} in Section~\ref{NumericalLyapunovSection} shows that, in fact, the error remains very small even for small values of $\n$.

\medskip

It is known that whenever the   random Lyapunov exponent  \eqref{RE} is positive, then it equals the {\em random entropy}, which is the random generalization of the Kolmogorov-Sinai entropy: this equality is called the {\em entropy formula}.  
Roughly speaking, the random entropy computes the limit of the Shannon entropy of the join partition generated by the successive application of the backward images of the random maps on an initial generating partition, aka the entropy rate.
We defer to \cite[Theorem~1.3]{kifier1986ergodic} for the precise definition  and \cite[Theorem~3.2]{araujo2005stochastic} for the connection with the RLE. 
What is important for us is that the random entropy  coincides with the much  easier  object which is the RLE when the latter is positive.

In Corollary~\ref{cor_pos_neg_Lexp} we proved that the Lyapunov exponent \eqref{Lexp} may be negative, and there is strong numerical evidence that the RLE \eqref{Lexp_rand_maps} is also negative for certain parameters; see Section~\ref{NumericalLyapunovSection}.
By assuming that the RLE is negative or zero, we then get that the random entropy is also zero, by using another important result connecting entropy and Lyapunov exponents, namely the Margulis-Ruelle inequality (see for instance \cite{ruelle1978inequality}), which states that the metric entropy is bounded by the maximum between zero and the sum of the positive Lyapunov exponents. 
The random version of this inequality, which we use,  has been proved by Kifer in \cite[Theorem~1.4]{kifier1986ergodic}.

We should now  stress the interesting fact that although the entropy is zero, the Markov chain \emph{mixes exponentially fast}, as we pointed out in Remark~\ref{cro}.  
This means that for any observables $f\in L^1$, $g\in BV$ the correlations $\int (\cL_{\n}^k f)(x) g(x) dx$ converge to $\int f(x) d\mu_{\n}(x)\int g(x) dx$ exponentially fast when $k\to\infty$.
This result can be stated in a more suggestive way by relying on random transformations. 
By using the notations introduced in Section~\ref{RandomTransformSection}, we can in fact rewrite the previous correlation  in terms of composition of randomly chosen maps and say that there exist $0<v<1$ and $C>0$, depending only on the system, such that, for all $k\ge0$,
$$
\left| \iint f(x) g(T_{\eta_k}\circ\dots\circ T_{\eta_1}(x)) d\overline{\eta}dx - 
\int f d\mu_{\n} \int g(x) dx \right|
\le C v^k \|f\|_1 \|g\|_{BV}.
$$
This exponential decay of correlations is a consequence of the spectral gap prescribed by the  quasi-compactness of the Markov operator proved in Proposition~\ref{ubv_fin_many_prop}, and of the uniqueness of the absolutely continuous  stationary measure, see \cite{bahsoun2014pseudo} for details; of course these properties hold even when the Lyapunov exponent is negative or zero.

\section{Numerical results}\label{nnrr} \label{NumericalSection}

In this section we describe and discuss some numerical experiments in support of our rigorous theoretical investigations. 
Specifically, we present the bifurcation diagram associated with the unimodal map \eqref{unimodal_map_rewritten} and compute the corresponding Lyapunov exponents for both the deterministic and the stochastic version of the map, see Subsections~\ref{subsec:bifurcation} and \ref{NumericalLyapunovSection}, respectively.\footnote{The code to reproduce all the figures and tables of this section are available from the corresponding author upon reasonable request.}

\subsection{Bifurcation diagram}\label{subsec:bifurcation}

In this subsection we analyse the dynamics of the unimodal map \eqref{unimodal_map_rewritten}.
The bifurcation diagram of a dynamical system shows how the asymptotic distribution of a typical orbit varies as a function of a parameter. 
For our map either the memory parameter $\omega$ or the parameter $\phi^*$ can be employed as bifurcation parameter. 
Fig.~\ref{fig:bif_diag} shows the bifurcation diagram as a function of $\phi^*$ (the bifurcation diagram as a function of $\omega$ looks similar). 
The choice of the parameters for this plot corresponds to a vertical segment in the parameter space (see Fig.~\ref{fig:dyncore}) with $\omega=0.5$ and $\phi^*$ varying in a neighbourhood of the dynamical core area. 
As explained in Section~\ref{UnimodalMapsSection}, Theorem~\ref{thmA}, the invariant set of the unimodal map  \eqref{unimodal_map_rewritten} could be an attracting periodic orbit, a Cantor set of measure zero or a finite union of intervals with a dense orbit, depending on the parameters $\omega$ and $\phi^*$. 
Specifically, there is an attracting fixed point or a 2-cycle outside the dynamical core region, whereas in the dynamical core the situation is more complex as small parameter variations can change the dynamics from chaotic to periodic and back, as we see in Fig.~\ref{fig:bif_diag}.  
To identify more precisely the signature of a chaotic behaviour, in the next subsection we compute the Lyapunov exponent as a function of $\phi^*$.

\begin{figure}
\centering
\includegraphics[width=1\linewidth]{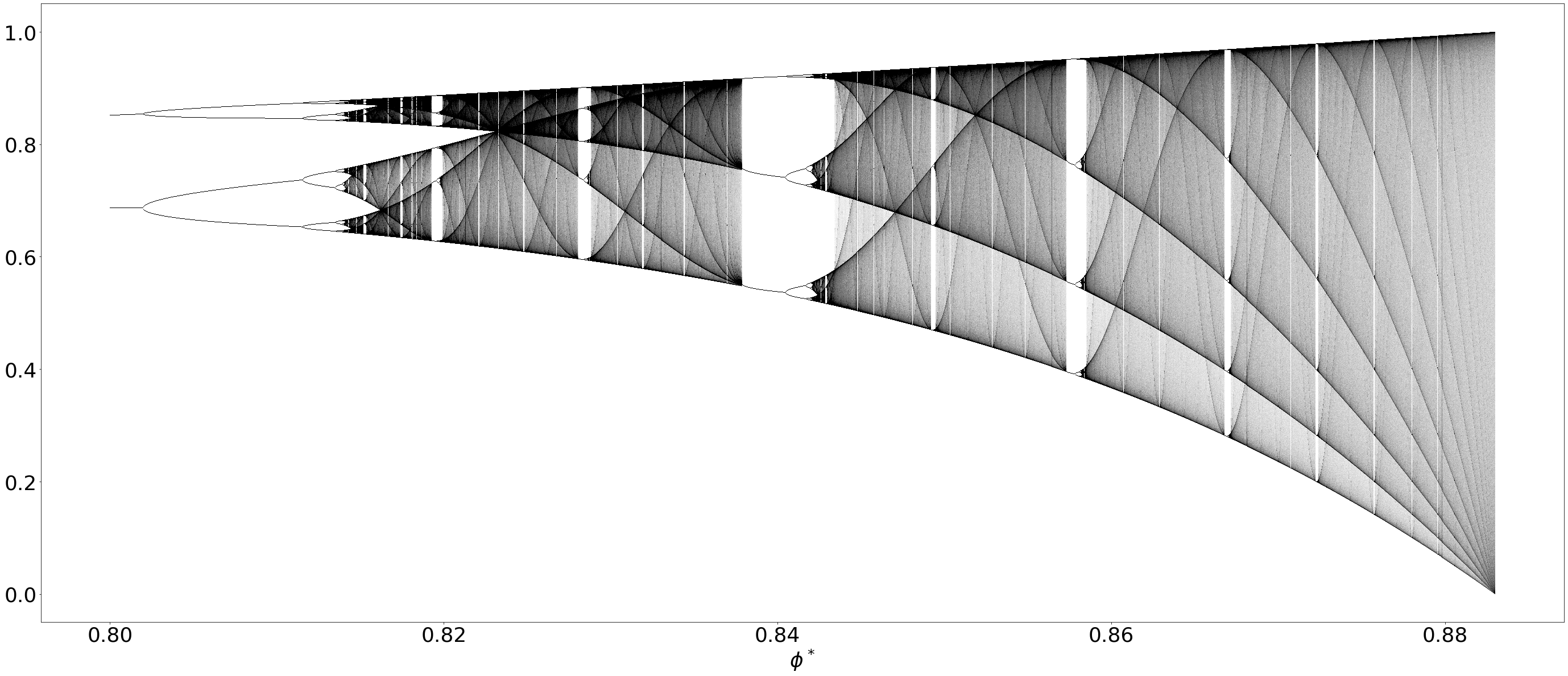}
\caption{Bifurcation diagram for $T$ in the dynamical core region ($\omega=0.5$).}
\label{fig:bif_diag}

\vspace*{\floatsep}

\includegraphics[width=1\linewidth]{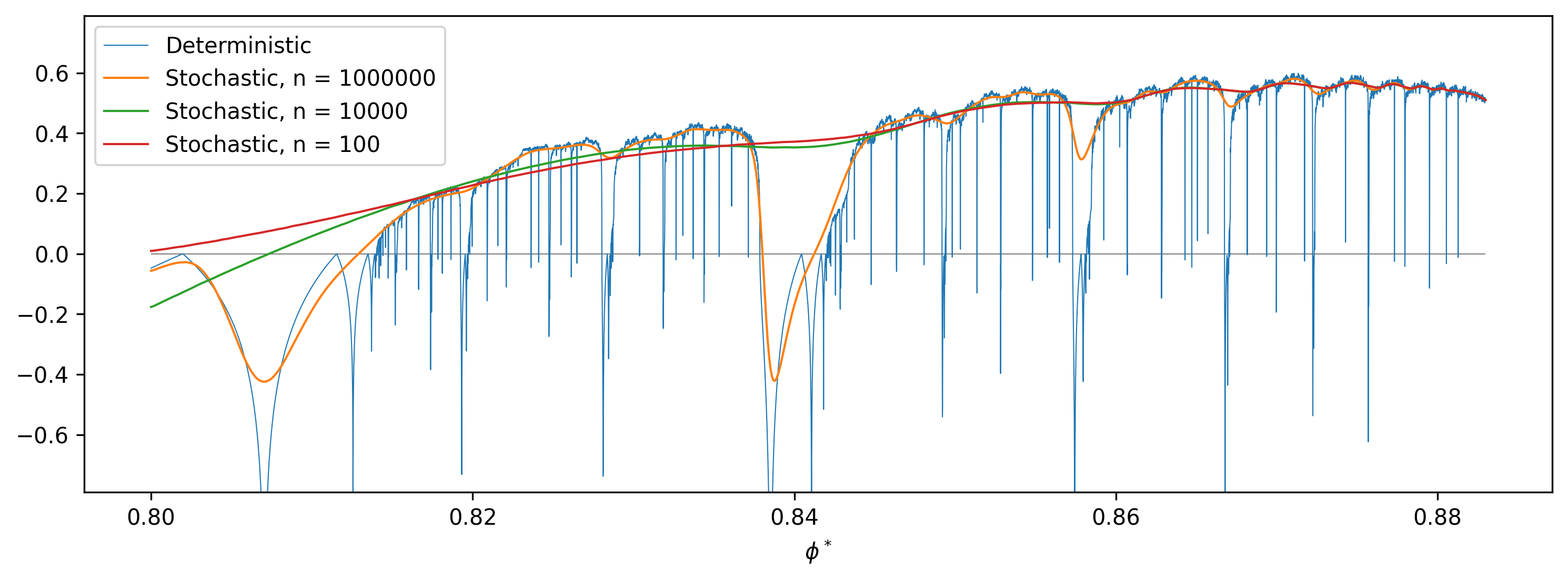}
\caption{Lyapunov exponents for deterministic and stochastic maps ($\omega=0.5$).}\label{f6}
\label{fig:lyap_exp_slice}

\vspace*{\floatsep}

\centering
\begin{subfigure}{.5\textwidth}
  \centering
  \includegraphics[width=1\linewidth]{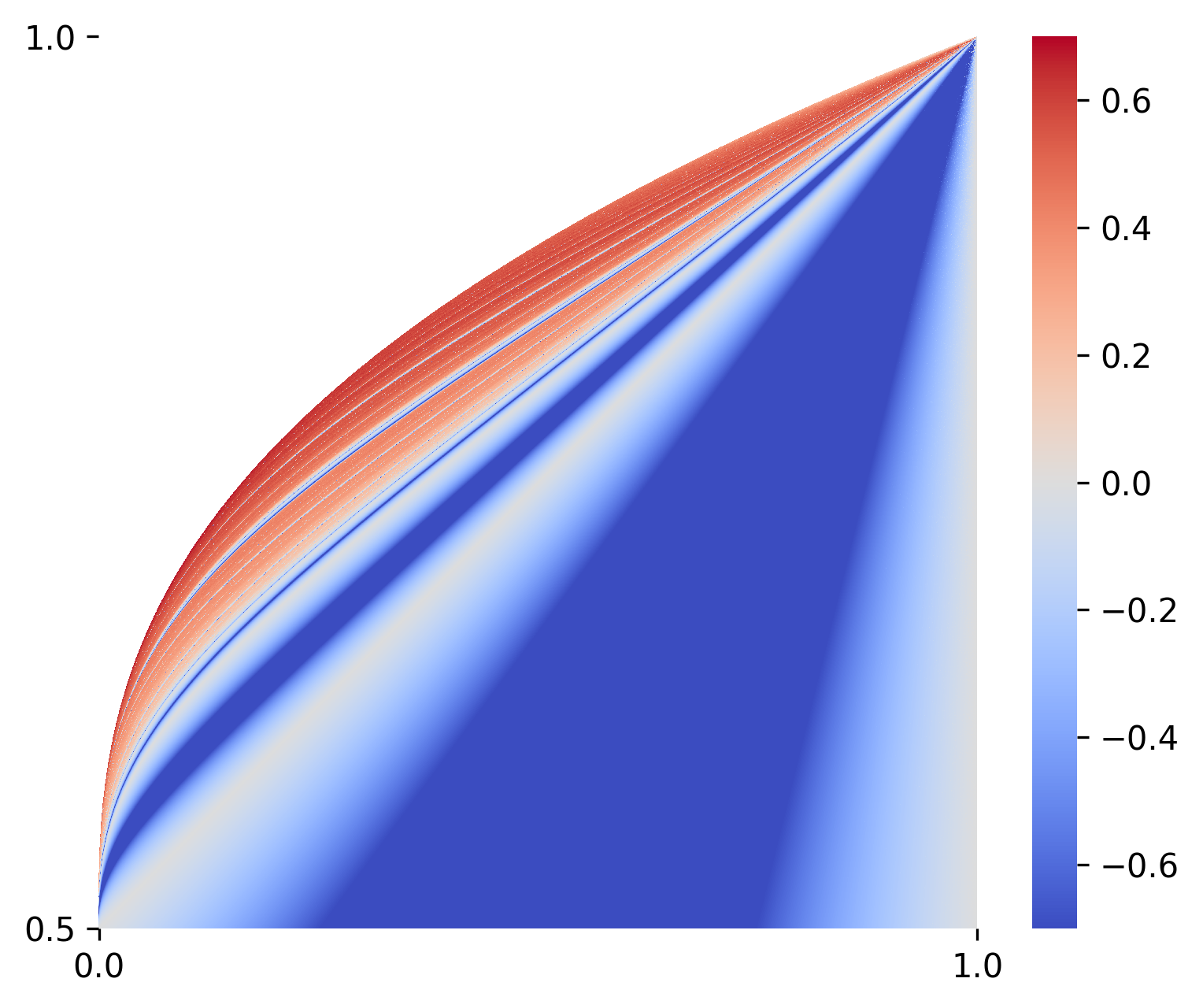}
  \put(-170,75){{\small $\omega$}}
  \put(-100,0){{\small $\phi^*$}}
  \caption{}
  \label{fig:lyap_det_plot}
\end{subfigure}%
\begin{subfigure}{.5\textwidth}
  \centering
  \includegraphics[width=1\linewidth]{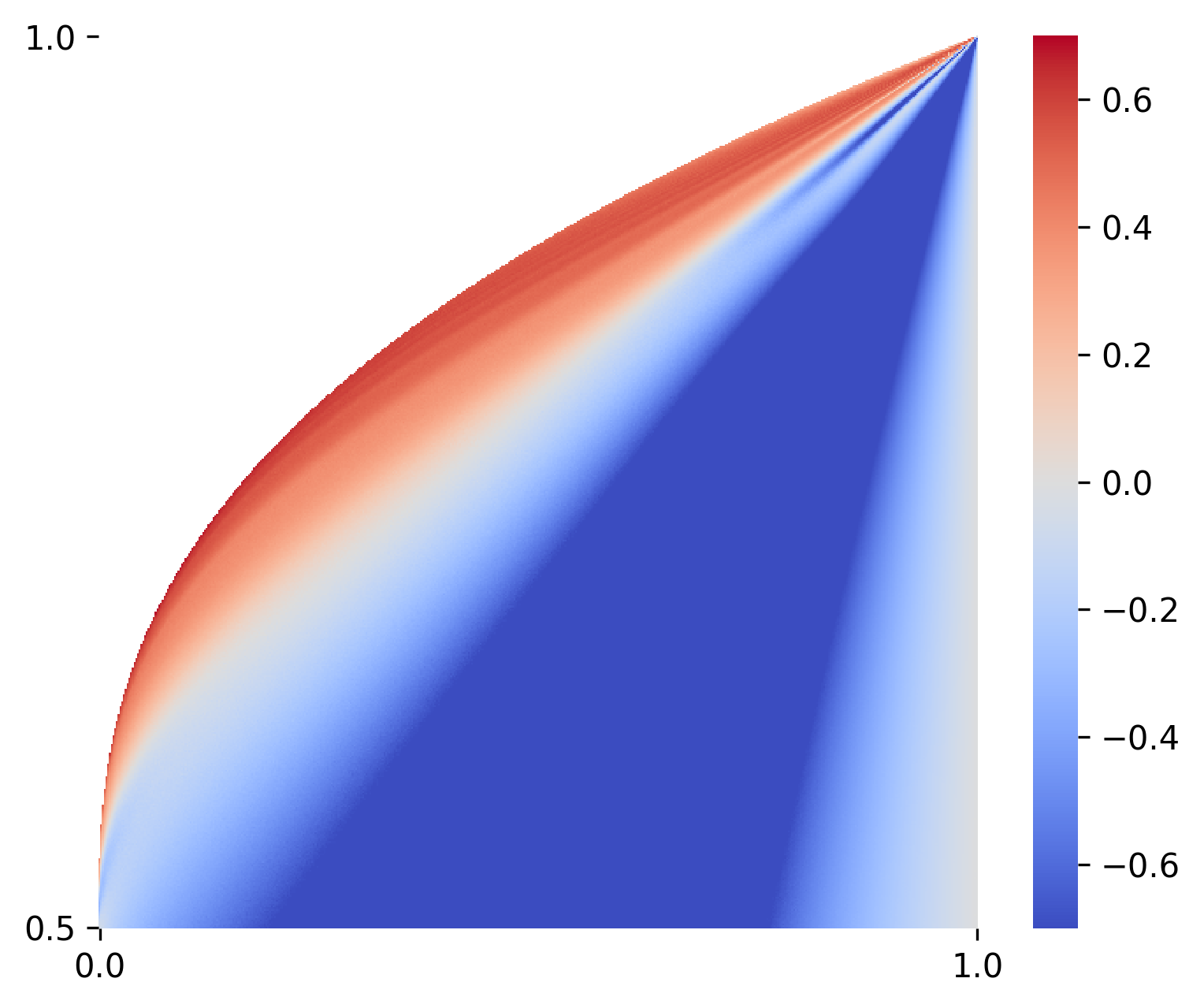}
  \put(-100,0){{\small $\phi^*$}}
  \caption{}
  \label{fig:lyap_stoc_plot}
\end{subfigure}
\caption{
Contour plot of the average Lyapunov exponents for (a)~deterministic and (b)~stochastic maps ($\n=1$).
}
\label{fig:lyap_plots}
\end{figure}

\subsection{Lyapunov exponent}\label{NumericalLyapunovSection}
The Lyapunov exponent for the deterministic map \eqref{unimodal_map_rewritten} is positive if and only if $T$ admits an absolutely continuous invariant measure, see Theorem~\ref{thmC} in Section~\ref{UnimodalMapsSection}. 
Fig.~\ref{fig:lyap_exp_slice} shows the estimated Lyapunov exponent in the same slice of the parameter space as in Fig.~\ref{fig:bif_diag}. 
Notice that the exponent becomes a smooth function of $\phi^*$ when we add even a very small noise, in agreement with the results of Section~\ref{LyapExpContinuitySection}. 
Fig.~\ref{fig:lyap_plots} shows the contour plot of the Lyapunov exponent as a function of $\phi^*$ and $\omega$ both for the deterministic map \eqref{unimodal_map_rewritten} and for the stochastic process described in Section~\ref{MainExampleSection}. 
The right panel shows a clear noise-induced regularization phenomenon.
In fact, for the stochastic version of the map the intricate fine structure in the parameter dependence of the Lyapunov exponent disappears and is replaced by a smooth dependence. 

To provide a numerical exemplification of the stochastic stability and of Theorem~\ref{Lexp_conv_n}, we computed average \eqref{Lexp} and random \eqref{Lexp_rand_maps} Lyapunov exponents, as well as the Lyapunov exponent for the deterministic map \eqref{unimodal_map_rewritten}, for different values of $\phi^*$, $\omega$ and $\n$.
The results are presented in Table~\ref{tab:lyap_exp_conv}.
Within each row, the two subrows are ALE and RLE, respectively: the two agree very well, in most cases up to the precision of the numerical computation.
In both cases we sampled 128 independent realizations of the process, each of length $10,000$.

\begin{table}[h]
\centering
\begin{tabular}{|c|c|c|c|c|c|c|c|c|}
\hline
\multirow{3}{*}{$\phi^*$} & \multirow{3}{*}{$\omega$} & \multirow{3}{*}{DC} & \multirow{3}{*}{Per.} & \multicolumn{5}{c|}{Lyapunov exponent}\\
\cline{5-9}
 & & & & \multicolumn{4}{c|}{$\n$} & \multirow{2}{*}{Det.}\\
\cline{5-8}
 & & & & $1$ & $10^3$ & $10^6$ & $10^9$ & \\
\hline
\multirow{2}{*}{0.845} & \multirow{2}{*}{0.557} & \multirow{2}{*}{yes} & \multirow{2}{*}{no} & 0.287 & 0.287 & 0.349 & 0.341 & \multirow{2}{*}{0.340} \\
 &  &  &  & 0.286 & 0.286 & 0.349 & 0.341 &  \\
\hline
\multirow{2}{*}{0.795} & \multirow{2}{*}{0.390} & \multirow{2}{*}{yes} & \multirow{2}{*}{no} & 0.345 & 0.346 & 0.389 & 0.398 & \multirow{2}{*}{0.400} \\
 &  &  &  & 0.345 & 0.346 & 0.389 & 0.399 & \\
\hline
\multirow{2}{*}{0.904} & \multirow{2}{*}{0.627} & \multirow{2}{*}{yes} & \multirow{2}{*}{no} & 0.557 & 0.557 & 0.560 & 0.550 & \multirow{2}{*}{0.552} \\
 &  &  &  & 0.558 & 0.557 & 0.560 & 0.550 & \\
\hline
\multirow{2}{*}{0.821} & \multirow{2}{*}{0.439} & \multirow{2}{*}{yes} & \multirow{2}{*}{yes} & 0.378 & 0.375 & $-0.051$ & $-0.159$ & \multirow{2}{*}{$-0.158$} \\
 &  &  &  & 0.378 & 0.375 & $-0.052$ & $-0.159$ & \\
\hline
\multirow{2}{*}{0.944} & \multirow{2}{*}{0.826} & \multirow{2}{*}{yes} & \multirow{2}{*}{yes} & 0.296 & 0.297 & 0.049 & $-0.123$ & \multirow{2}{*}{$-0.122$} \\
 &  &  &  & 0.297 & 0.296 & 0.052 & $-0.123$ & \\
\hline
\multirow{2}{*}{0.766} & \multirow{2}{*}{0.323} & \multirow{2}{*}{yes} & \multirow{2}{*}{yes} & 0.320 & 0.324 & 0.286 & $-0.076$ & \multirow{2}{*}{$-0.046$} \\
 &  &  &  & 0.320 & 0.325 & 0.286 & $-0.076$ & \\
\hline
\multirow{2}{*}{0.258} & \multirow{2}{*}{0.837} & \multirow{2}{*}{no} & \multirow{2}{*}{yes} & $-0.243$ & $-0.248$ & $-0.248$ & $-0.248$ & \multirow{2}{*}{$-0.248$} \\
 &  &  &  & $-0.286$ & $-0.248$ & $-0.248$ & $-0.248$ & \\
\hline
\multirow{2}{*}{0.908} & \multirow{2}{*}{0.804} & \multirow{2}{*}{no} & \multirow{2}{*}{yes} & $-0.284$ & $-0.285$ & $-0.365$ & $-0.362$ & \multirow{2}{*}{$-0.362$} \\
 &  &  &  & $-0.286$ & $-0.287$ & $-0.366$ & $-0.362$ & \\
\hline
\multirow{2}{*}{0.541} & \multirow{2}{*}{0.227} & \multirow{2}{*}{no} & \multirow{2}{*}{yes} & $-0.619$ & $-0.441$ & $-0.380$ & $-0.380$ & \multirow{2}{*}{$-0.380$} \\
 &  &  &  & $-0.578$ & $-0.446$ & $-0.380$ & $-0.380$ & \\
\hline
\end{tabular}
\caption{Average and random Lyapunov exponents for different values of $\n$ compared to the Lyapunov exponent for the deterministic map.
DC stands for dynamical core, Per.\ for periodic, and Det.\ for the Lyapunov exponent of the deterministic map.}
\label{tab:lyap_exp_conv}
\end{table}

\section{Estimating the map parameters via deep neural networks}\label{sec:estimationDNN}

We now consider the problem of estimating the parameters of the map from (short) time series.
This is motivated by the fact that in Section~\ref{sec:empirical} we empirically investigate a dataset of US commercial banks leverage.
We will consider the time series of bank's leverage as realizations of the process described in Section~\ref{MainExampleSection} and will estimate for each bank the model parameters.
Each time series is very short, being composed by only 59 points.  

Given the random nature of the map, one could use maximum likelihood estimation to estimate the parameters. 
However, this approach is not feasible for two reasons. 
First, the likelihood function is highly non-convex, so that standard optimization methods may perform poorly. 
Second, although the likelihood function for the process itself can be written explicitly, in many cases it may happen that the observed time series are systematically undersampled and this prevents an explicit calculation of the likelihood function.
For example, we may observe only one slow time scale, corresponding to portfolio rebalancing, out of two, or even out of three (i.e.\ the bank's risk assessment and portfolio composition may be updated more frequently than our quarterly observations, for instance at a monthly frequency). 
If we observe, for instance, only the second iterate of the process, 
\begin{equation*}
    \phi_{t+2} = T(T(\phi_t;\theta);\theta) + \sigma(\phi_t;\theta)\epsilon_t) + \sigma(\phi_{t+1};\theta)\epsilon_{t+1},\quad t\in\mathbb{Z},
\end{equation*}
the transition probabilities $p(\phi_{t+2} | \phi_{t};\theta)$, $t\in\mathbb{Z}$, are no longer Gaussian (as it would be the case if we observe the first iterate). 
Hence, there is apparently no effective formula for the likelihood function.
 
For this reason, to estimate the parameters of the map, we propose to use a convolutional neural network (CNN) consisting of a sequence of convolutional layers followed by a sequence of dense, or fully connected, layers; we refer to \cite[Chapter~9]{goodfellow2016deep} for detailed exposition of CNNs.
In order to deal with the possibility that the observed time-series are realizations of certain iterates of the process, we separately optimize two CNN architectures to be used sequentially. 
First, we optimize a CNN (henceforth denoted by $\mathsf{CNN1}$) for estimating the number of iterates between observations: it takes as input time series of length 59 and outputs the corresponding value $k$ of the map's iterate that generated the time-series. 
Second, for each value of $k$, we optimize a CNN (henceforth denoted by $\mathsf{CNN2}(k)$), having the exact same inputs, to output the corresponding parameters $(\phi^*, \omega)$ that generated the time-series.
Once the parameters $(\phi^*, \omega)$ are estimated, the variance of the noise (and therefore $\n$) can be estimated by standard methods.
To train $\mathsf{CNN1}$ and $\mathsf{CNN2}(k)$ we used a training set of one million samples simulated from the model in Section~\ref{MainExampleSection} with values of the parameters  $\theta = (\phi^*, \omega, \n)$ which uniformly span the parameter space.
For both steps, when simulating the series, the initial state of the system was taken randomly from a uniform distribution on $[0,1]$. 
This is especially important because of the relatively short length of the series at hand.
Therefore, being based only on simulations, the NN approach, contrary to the maximum likelihood one, can work also for partial observations.

The architectures of $\mathsf{CNN1}$ and $\mathsf{CNN2}(k)$ are schematized in Fig.~\ref{fig:cnn1}.
In general, a convolutional layer is composed of $n_f$ filters and each filter is associated with one kernel that is applied to a small moving window of the time-series; for instance, in our first convolutional layer $n_f = 128$ and all the windows are of width $2$.
The outputs of one convolutional layer are connected to the next layer.
The weights of these connections constitute the NN parameters to be optimized. 
After seven convolutions, the output is passed to a sequence of dense layers, which concludes the NN. 
We use the rectified linear unit (ReLU) function as activation function.
Now, we describe the experimental setup.

\begin{figure}[ht]
\centering
{\scriptsize
\begin{BVerbatim}
Model CNN1: "convolutional_categorical_model"
_________________________________________________________________
Layer (type)                 Output Shape              Param #   
=================================================================
reshape (Reshape)            (None, 59, 1)             0         
_________________________________________________________________
conv1d_1 (Conv1D)            (None, 58, 128)           384       
_________________________________________________________________
conv1d_2 (Conv1D)            (None, 29, 64)            16448     
_________________________________________________________________
conv1d_3 (Conv1D)            (None, 15, 64)            8256      
_________________________________________________________________
conv1d_4 (Conv1D)            (None, 8, 64)             8256      
_________________________________________________________________
conv1d_5 (Conv1D)            (None, 4, 64)             8256      
_________________________________________________________________
conv1d_6 (Conv1D)            (None, 2, 64)             8256      
_________________________________________________________________
conv1d_7 (Conv1D)            (None, 1, 64)             8256      
_________________________________________________________________
flatten (Flatten)            (None, 64)                0         
_________________________________________________________________
dense_1 (Dense)              (None, 128)               8320      
_________________________________________________________________
dense_2 (Dense)              (None, 64)                8256      
_________________________________________________________________
dense_3 (Dense)              (None, 3)                 195       
=================================================================
Trainable params: 74,883

Model CNN2: "convolutional_model"
...  
_________________________________________________________________
dense_3 (Dense)              (None, 2)                 130       
=================================================================
Trainable params: 74,818
\end{BVerbatim}
}
\caption{Architectures of the $\mathsf{CNN1}$ model used to estimate the iterate $k$ and the $\mathsf{CNN2}(k)$ model used to estimate the parameters $(\phi^*,\omega)$ for each $k$.
The two models differ only in the output layer.}
\label{fig:cnn1}
\end{figure}

The implementation is carried out in Python. 
To generate training and testing data we simulate one million samples. 
$\mathsf{CNN1}$ is optimized with the stochastic gradient descent method by using the Adam algorithm \cite{kingma2014method}, the categorical cross-entropy as loss function, and the accuracy as metric, with $L^2$ regularization of weights.
Instead, to optimize $\mathsf{CNN2}$ we use the Mean Squared Error (MSE) both as loss function and as metric. 
The batch size is $32$ in both cases.
The seven convolutional and three dense layers have a total of $74,818$ trainable parameters.

The CNN models show a good performance. 
We tested our methods on a testing set of 100,000 out of sample time series. 
Fig.~\ref{fig:conf_matr_k} shows the accuracy of $\mathsf{CNN1}$ to estimate the iterates on test data.
We choose $k = 1, 2, 3$ because of our empirical application of Section~\ref{sec:empirical}. 
The MSE of $\mathsf{CNN2}(k)$ on the test set is about 0.001 for each $k$. 
Since both $\phi^*$ and $\omega$ are uniformly distributed in $[0,1]$, the MSE is quite small and the NN effective.

\begin{figure}[h]
\centering
\includegraphics[width=0.4\linewidth]{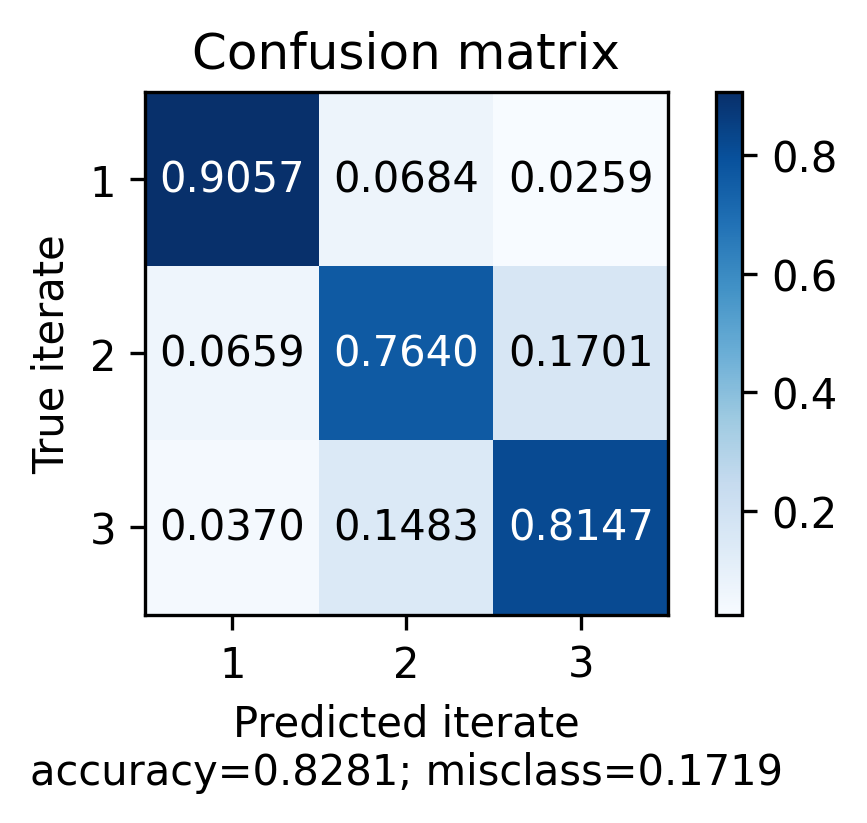}
\caption{
Accuracy of the $\mathsf{CNN1}$ model used to estimate the iterates.
}
\label{fig:conf_matr_k}
\end{figure}

\section{Chaos in real bank leverage time series}\label{sec:empirical}

In this section we perform an empirical analysis of a large set of bank's leverage time series. 
We first describe the the data set. 
Next, through $\mathsf{CNN1}$ and $\mathsf{CNN2}(k)$, we estimate the observed iterate $k$ as well as the parameters $(\phi^*, \omega)$ of the model and discuss the results, investigating the relation between the estimated parameters and the bank's size.
Finally, to perform an independent analysis supporting our conclusion, we apply the Chaos Decision Tree Algorithm \cite{toker2020simple} to these time-series and compare the resulting classification with the one obtained with NN estimates.

\subsection{Data set}\label{subsec:dataset}

We use the data set of US \textit{Commercial Banks and Savings and Loans Associations} provided by the \textit{Federal Financial Institutions Examination Council} (FFIEC). 
For the sake of completeness, we provide here a description of it, referring to \cite{di2018assessing} and references therein for further details. 
A Commercial Bank is defined\footnote{See \url{http://www.ffiec.gov/nicSearch/FAQ/Glossary.html}.} officially by the \textrm{FFIEC} as: \textit{``[\ldots] a financial institution that is owned by stockholders, operates for a profit, and engages in various lending activities"}. 
Commercial banks quarterly fill the \textit{Consolidated Report of Condition and Income} (generally referred to as \textit{Call Report}) as required by the \textrm{FFIEC}.
A Savings and Loan Associations, instead, is a financial institution that accepts deposits primarily from individuals and channels its funds primarily into residential mortgage loans. 
Starting from the first quarter of 2012 they are required to file the same reports of Commercial Banks, thus they are included in the data set since then. 
The data provided by the \textit{Call reports} are publicly available since 1986 although the level of details required has increased over time. 
To have a good compromise between the fine structure of data and a reasonably populated statistics we follow \cite{di2018assessing} and consider the time period going from March 2001 to December 2014, for a total of 59 quarters.  
Also we consider only the financial institutions that are present in the data set in all the quarters for a total of $5,031$ banks.
The financial leverage $\lambda_t$ of each institution at time $t$ is defined as the ratio between the sum of its assets and its equity at time $t$. 
In particular, the latter is given by $E_t=A_t-L_t$ where $L_t$ represents the liabilities and $A_t$ the assets of the bank, thus $\lambda_t = {A_t}/{E_t}$.

\subsection{Estimation via neural networks}

In order to estimate the parameters of the map on the just-described data set, we need to fix the value of the liquidity parameter $\gamma$; remind that we consider the linear transformation $\phi_t = (\lambda_t-1)/\gamma$.
In this work, we assume that the liquidity parameter $\gamma$ of the risk investment is the same for all the banks in our data set.
Admittedly, this is a simplifying assumption, coherent with the so-called assumption of statistical equivalence for risky investments (see also \cite{mazzarisi2019panic}), which allows for an analytical tractability of the model; a complete exploration of a relaxation of this hypothesis is beyond the scopes of the present paper and is, therefore, left for future work.
In order to fix its value, we exclude $662$ time series (out of $5,031$) that  contain outliers, which we define to be values that are two standard deviations away from the mean.
We then set $\gamma$ to the maximum over the remaining $4,369$ series, obtaining $\gamma = 15.969$.

Since the time series that we analyze contain quarterly data and portfolio decisions may be made more frequently, it is natural to assume that the observed time series are realizations of certain iterates of the process; we assume $k \in \{1, 2, 3\}$. 
Fig.~\ref{subfig:iterate_size_a} displays the output of $\mathsf{CNN1}$.
It turns out that only a small percentage (about $1\%$) of the banks in our data set rebalance their portfolios at a quarterly frequency.
Most banks seem to rebalance either every six weeks ($k=2$, about $55\%$) or every month ($k=3$, about $43\%$). 
One may ask if the portfolio re-balancing frequency is related to the size of the bank (defined as the average across the 59 quarters of the sum of the dollar-amount of all the type of assets detained by it), for example because larger banks manage more actively their portfolio.
Fig.~\ref{subfig:iterate_size_b} shows the box plots of the logarithm of the size of the banks for $k=1,2,3$.
We observe that there is not a statistically significant difference among them.

\begin{figure}[h]
\centering
\begin{subfigure}{.438\textwidth}
  \includegraphics[width=1\linewidth]{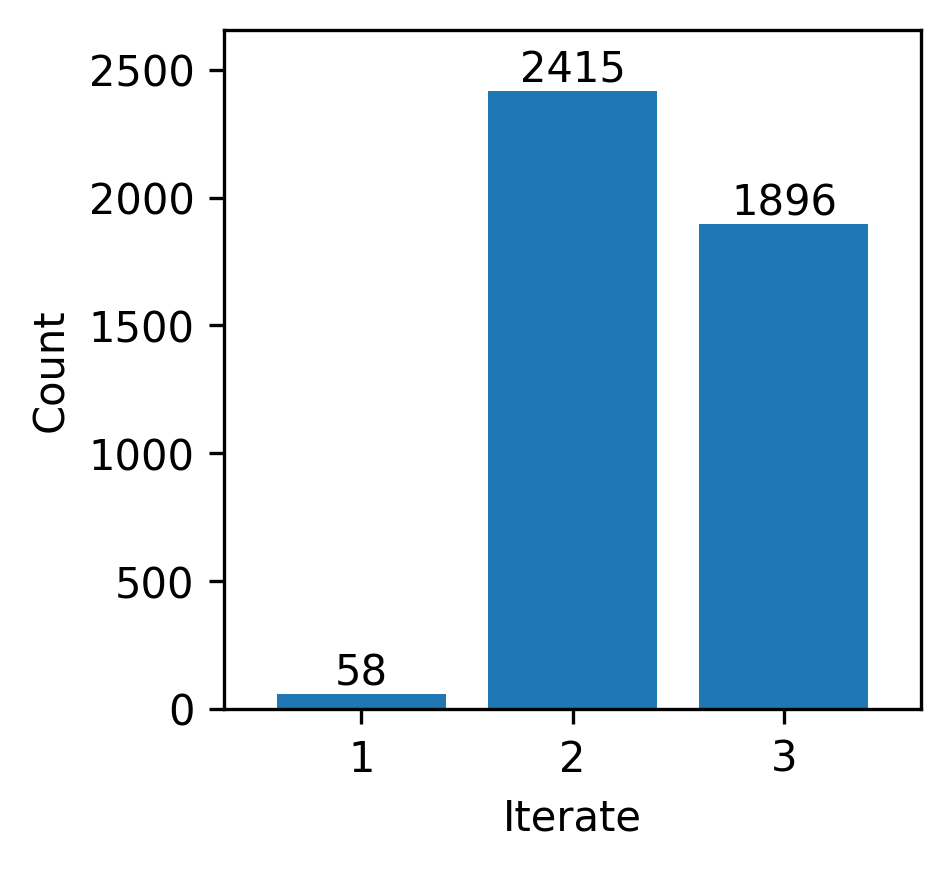}
  \caption{Number of banks by iterate.}
  \label{subfig:iterate_size_a}
\end{subfigure}
\begin{subfigure}{.4\textwidth}
  \includegraphics[width=1\linewidth]{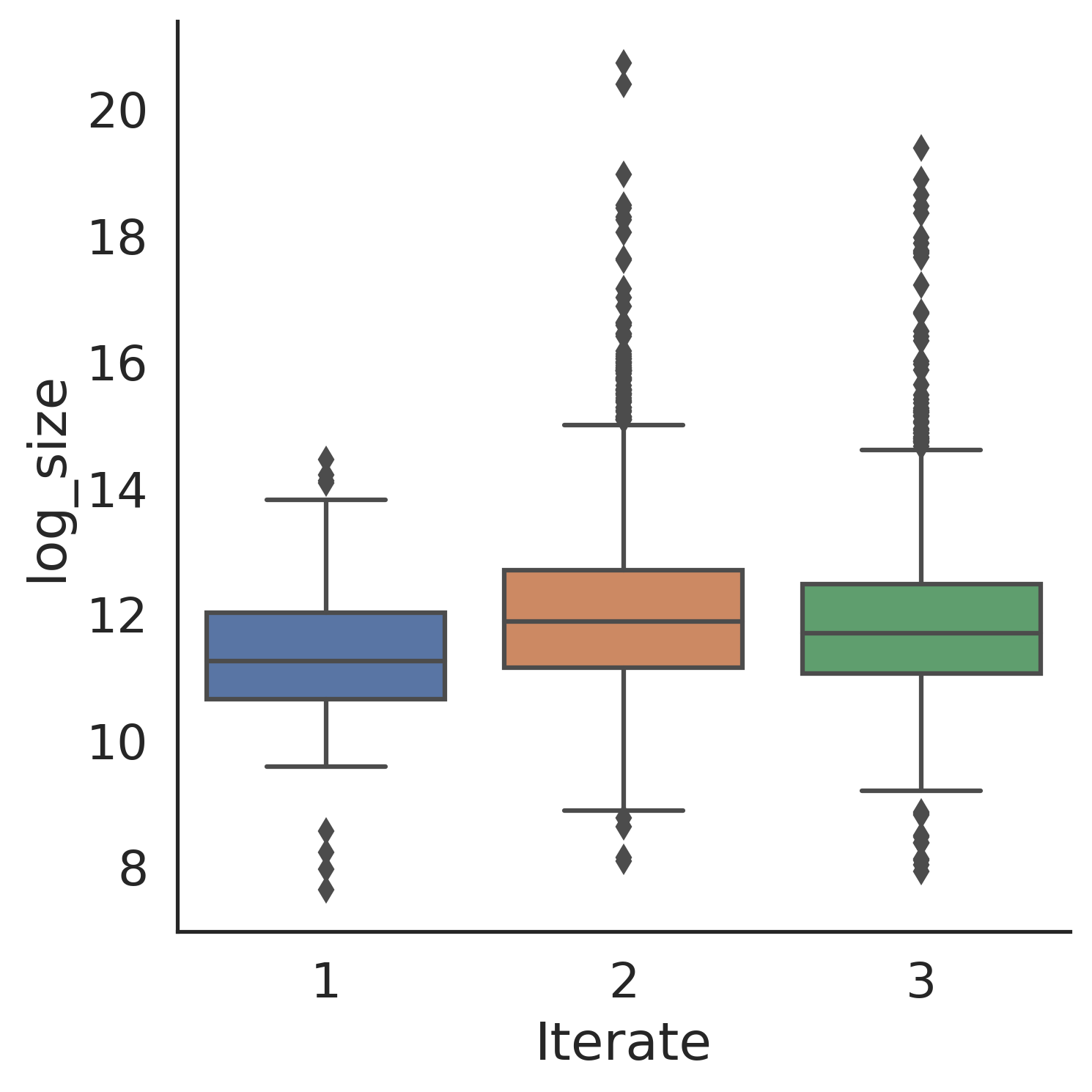}
  \caption{Log bank size by iterate.}
  \label{subfig:iterate_size_b}
\end{subfigure}
\caption{Results on the iterate of the process.}
\label{fig:iterate_size}
\end{figure}

Once the number of iterates $k$ has been identified, we proceed to the study of the chaotic behaviour of the time series. 
We divide the banks in the three groups identified by $k$ and employ $\mathsf{CNN2}(k)$ in order to estimate the parameters $(\phi^*,\omega)$. 
In Fig.~\ref{subfig:parameters_estimates_1_a} we plot the estimates of $\phi^*$ against those of $\omega$; pairs belonging to the dynamical core region are displayed in red, whereas those falling outside the dynamical core region are displayed in blue.
Interestingly, the percentage of banks for which the estimates $(\phi^*,\omega)$ are in the dynamical core region is about $12\%$.
Moreover, Fig.~\ref{subfig:parameters_estimates_1_b} indicates that $k$ is very often equal to two for these banks. 

\begin{figure}[h]
\centering
\begin{subfigure}{.5\textwidth}
  \includegraphics[width=1\linewidth]{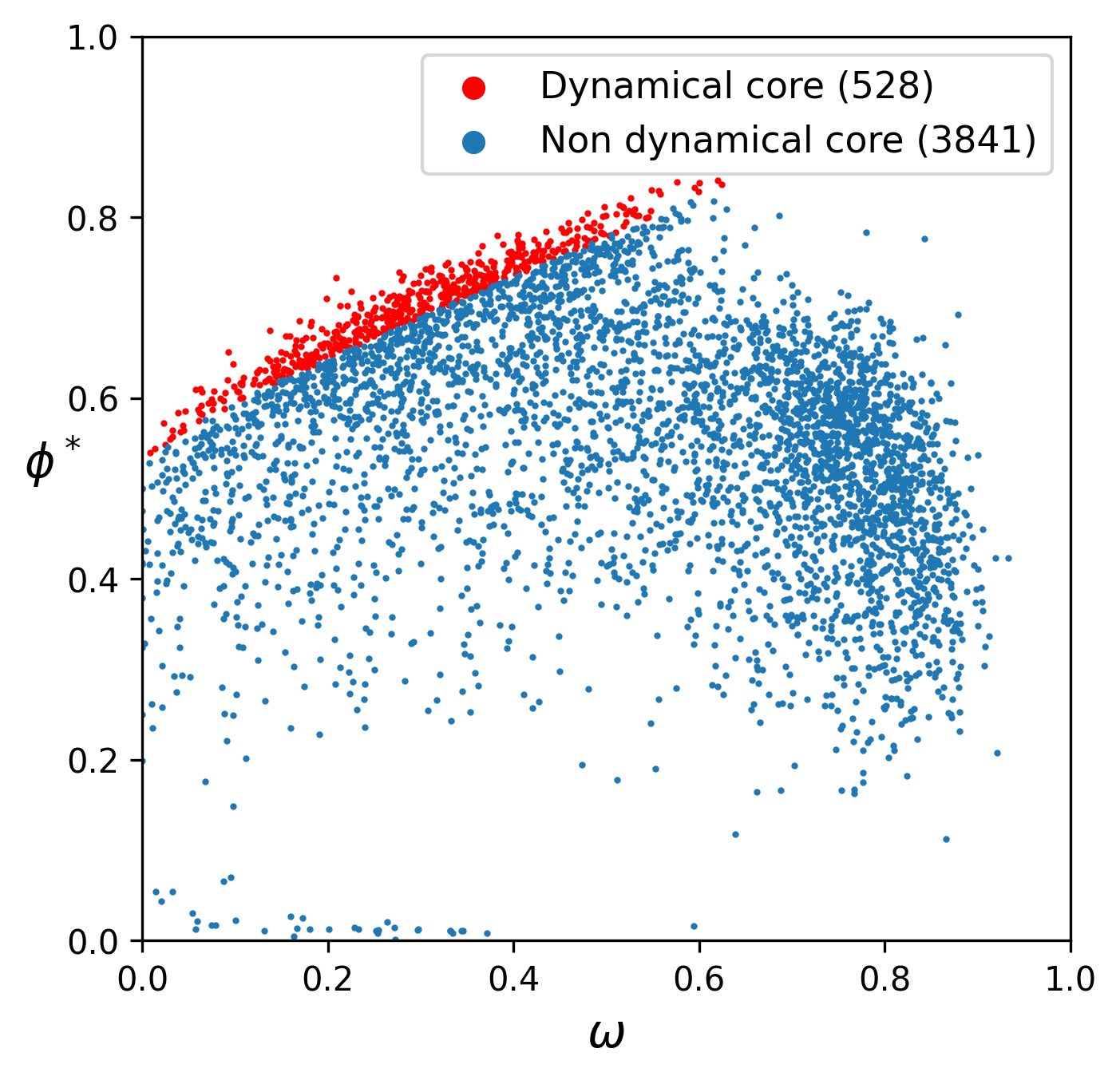}
  \caption{Parameters distribution by DC.}
  \label{subfig:parameters_estimates_1_a}
\end{subfigure}%
\begin{subfigure}{.5\textwidth}
  \includegraphics[width=1\linewidth]{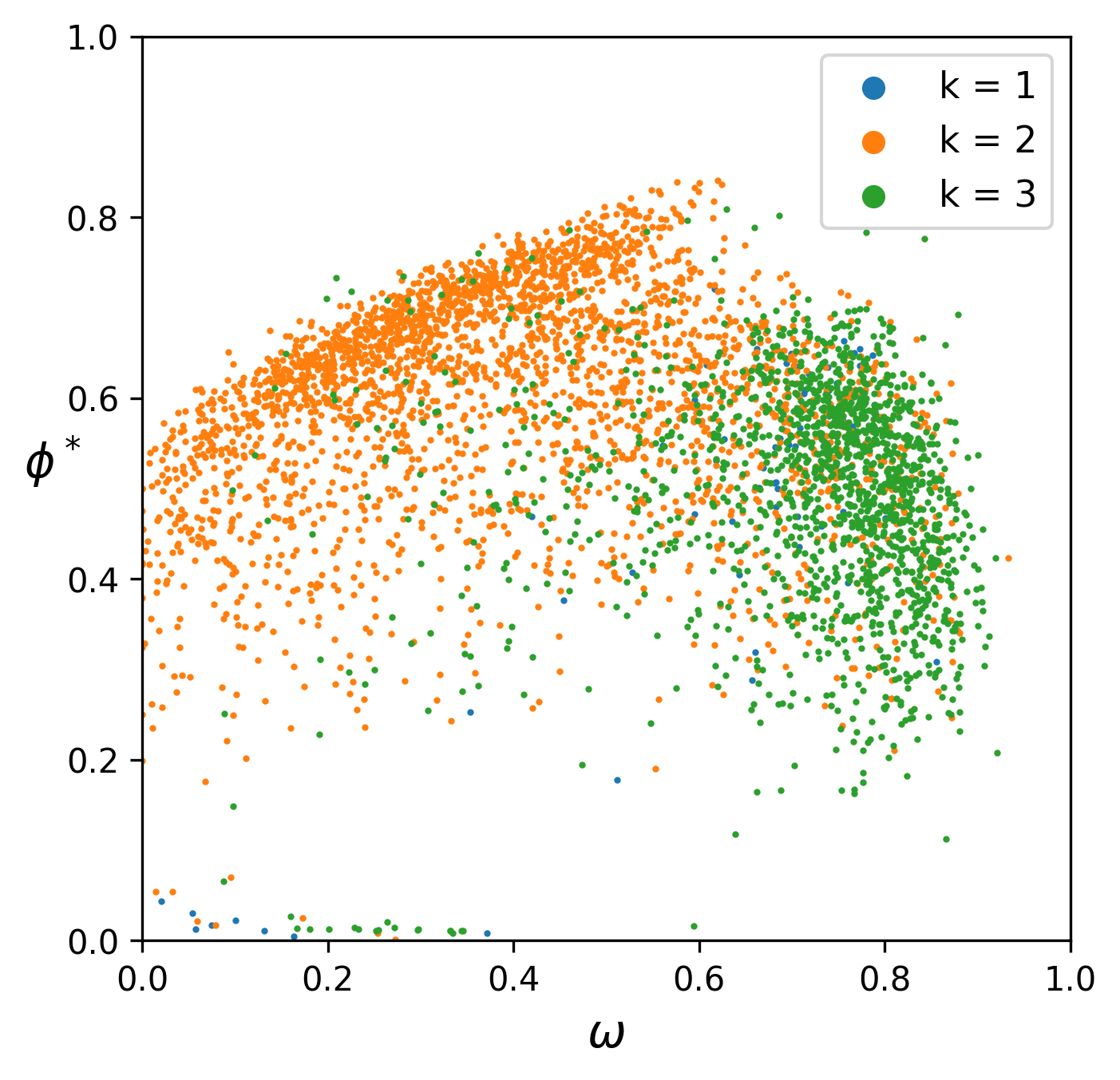}
  \caption{Parameters distribution by iterate.}
  \label{subfig:parameters_estimates_1_b}
\end{subfigure}
\caption{
Estimated parameters for iterates in $\{1,2,3\}$.
}
\label{fig:parameters_estimates_1}
\end{figure}
We now ask if bank size is related with the fact that the estimated pair $(\phi^*,\omega)$ is or not in the dynamical core region. 
Fig.~\ref{subfig:parameters_estimates_bank_size_a} shows the probability density functions of the logarithm of the banks size, by considering separately banks inside and outside the dynamical core, and Fig.~\ref{subfig:parameters_estimates_bank_size_b} displays the corresponding probability-probability plot.
To test that the difference between the distribution of banks sizes in and outside the dynamical core region is statistically significant we perform the Kolmogorov-Smirnov test of the null hypothesis that the two samples have the same distribution. 
The statistics of the test is $9.31 \times 10^{-2}$ corresponding to a $p$-value of $5.9 \times 10^{-4}$.
This latter value shows that the two subsamples have different distributions.

Summarizing, we have found that the parameters of a sizable fraction of banks lie in the dynamical core region and that the dynamics of the leverage of the larger banks tends to be more frequently in the dynamical core than that of the smaller banks.

\begin{figure}[h]
\centering
\begin{subfigure}{.5\textwidth}
    \includegraphics[width=1\linewidth]{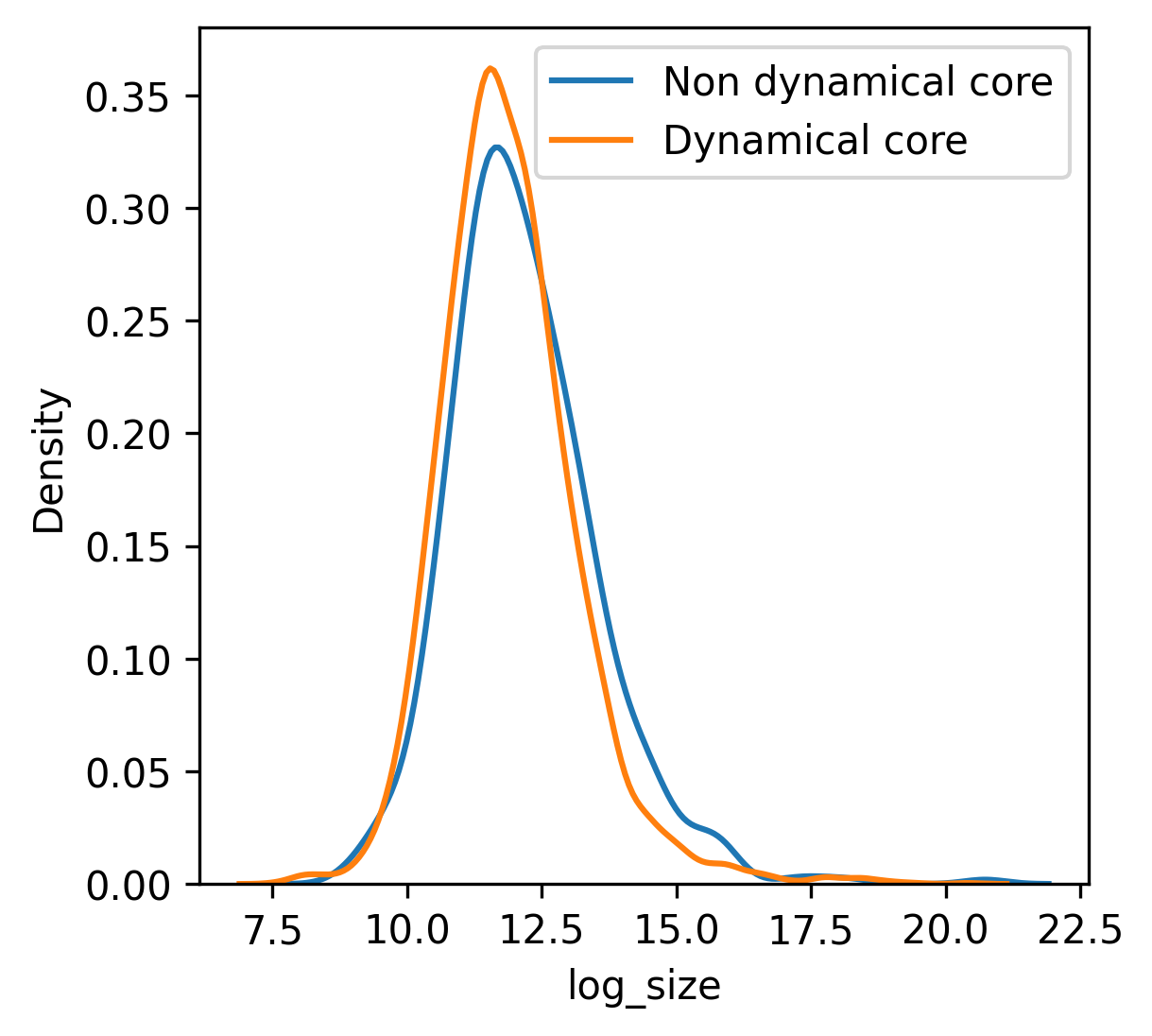}
    \caption{Probability density functions.}
    \label{subfig:parameters_estimates_bank_size_a}
\end{subfigure}%
\begin{subfigure}{.48\textwidth}
    \includegraphics[width=1\linewidth]{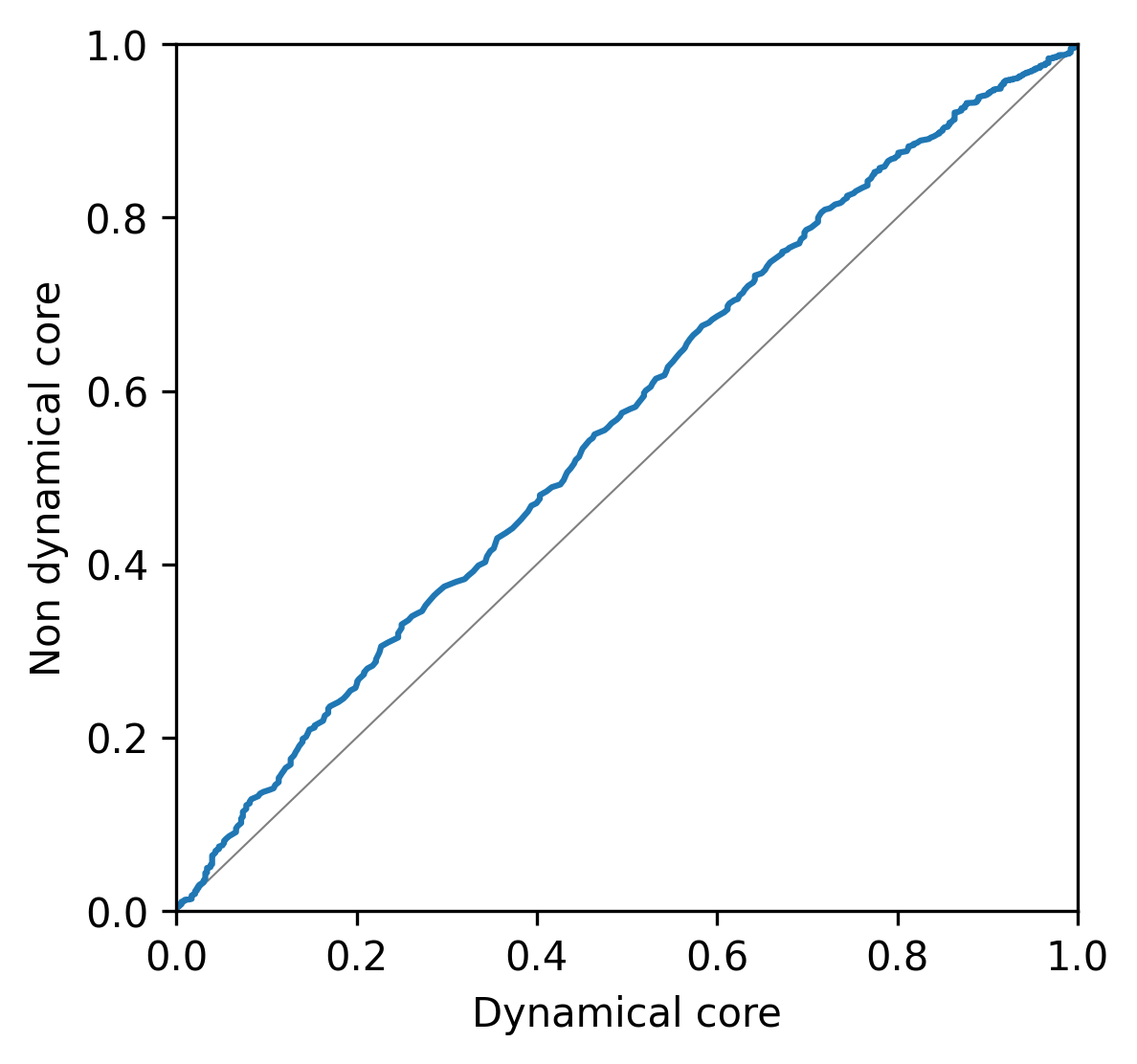}
    \caption{Probability-probability plot.}
    \label{subfig:parameters_estimates_bank_size_b}
\end{subfigure}
\caption{Results on the parameters of the process and bank size.}
\label{fig:parameters_estimates_bank_size}
\end{figure}

\subsection{Classification via the Chaos Decision Tree Algorithm}\label{subsec:dta_empirical}

Finally, we perform an independent analysis on the bank's leverage time series by making use of the recently proposed Chaos Decision Tree Algorithm (CDTA) \cite{toker2020simple}, described in detail in Appendix~\ref{app:cdta}. 
This is a non-parametric method which classifies an input time series as chaotic, periodic, or stochastic\footnote{Notice that the definitions of \emph{chaotic} and \emph{periodic} for the CDTA differ from the ones given in Section~\ref{UnimodalMapsSection}, see Appendix~\ref{app:cdta}.}.
We perform this analysis for two reasons. 
First, we know that chaotic behavior can be present only for series generated by our map with parameters in the dynamical core. 
Thus we test whether the series classified as chaotic by CDTA have estimated parameters in the dynamical core. 
The second reason is to count how many banks in the dynamical core are identified as chaotic or periodic by CDTA. 
The Appendix also contains the results of the application of CDTA to data simulated by our map for different time series length, level of noise $\n$, and number of iterates $k$. 

Applying CDTA, we find that $64\%$ of the banks are classified as stochastic, $\sim 12\%$ as periodic and $\sim 23\%$ as chaotic. 
The consistency between the classification made by CDTA and the partition  `dynamical core' and `not dynamical core' of the parameters space $(\phi^*, \omega)$ found by the NN model can be assessed by looking at Table~\ref{tab:dc_dta}.
We find that a large fraction of series outside the dynamical core are classified as stochastic by CDTA, while a third of banks in the dynamical core is classified as chaotic.
This fraction is significantly smaller for banks outside the dynamical core. 
Thus, despite the agreement is not perfect, we find a reasonable consistency between the conclusions of the two methods and, more importantly, find significant (and independent) support to the conclusion that a sizable fraction of bank time series are described by a chaotic dynamics.

\begin{table}[h]
\centering
\begin{tabular}{|l|c|c|c|}
\hline
& Periodic & Chaotic & Stochastic \\
\hline
Non dynamical core & 382 (9.98\%) & 648 (16.93\%) & 2798 (73.09\%) \\
\hline
Dynamical core & 107 (20.34\%) & 176 (33.46\%) & 243 (46.20\%) \\
\hline
\end{tabular}
\caption{Number of banks by classes.}
\label{tab:dc_dta}
\end{table}

Finally, the findings reported in the previous subsection suggest a positive relation between the size of a bank and the probability that the dynamics of the (corresponding) leverage time series is chaotic. 
To verify this observation by using the CDTA classification, we first rank the banks in quintiles according to their size and within each quintiles we compute the percentage of banks that are detected to be stochastic, periodic and chaotic. 
Table~\ref{tab:Size_statistics} collects the results.
In a nutshell, banks having a larger size have, on average, a larger percentage of leverage time series detected as chaotic. 
A $\chi^2$-test applied to contiguous quintiles rejects the hypothesis of independence of the CDTA classification from the quintile, indicating that the difference in frequencies across quintiles are statistically significant. 
Thus also the CDTA analysis confirms that larger banks are more likely characterized by chaotic time series of leverage.

\begin{center}
\begin{table}[ht]
\begin{center}
\begin{tabular}{|c|c|c|c|c|c|}
\hline 
Statistics & q$_{1}$ & q$_{2}$ & q$_{3}$ & q$_{4}$ & q$_{5}$ \\
\hline 
Chaotic (\%) & 18.2 & 20.2 & 21.7 & 23.6 & 29.1 \\
\hline
Periodic (\%) & 12.9 & 11.6 &  12.9 & 13 & 11.8 \\
\hline
Stochastic (\%) & 68.6 &  67.3 &  65 & 63 & 58.9 \\
\hline
\end{tabular}
\end{center}
\hspace{-1.5cm}
\caption{Fraction of banks classified as chaotic, periodic, or stochastic by CDTA conditionally to the decile of the bank size.}
\label{tab:Size_statistics}
\end{table}
\end{center}

\section{Conclusions}\label{sec:conclusion}

Most risk management practices (as, for example, Value-at-Risk) assume that prices are not affected by actions of other financial institutions that are managing the risk of their portfolio.
In other words, these practices assume that risk is exogenous. 
In reality, in the presence of limited liquidity, coordinated and homogeneous risk management can create market instability and result in what is known as endogenous risk.
This has the potential to amplify market instabilities and create crashes through the well-known feedback between leverage, risk, and asset prices.
An additional, and less considered, feedback between past and future risks is present because financial institutions use historical data to estimate  both  the riskiness  of their  investments  and  their correlations.
This creates new threats for the systemic stability of financial markets.
Studying how these two feedbacks affect the leverage dynamics is of paramount importance for understanding systemic risk.

In this paper we consider a stylized model where both feedbacks are present.
We showed that the dynamics of the bank's leverage is described by a unimodal map on $[0,1]$ perturbed with additive and heteroscedastic noise. 
The perturbed system can be described in two equivalent ways as a stationary Markov chain or in terms of random transformations. 
In both cases a fundamental  object  is the {\em stationary measure} of the process which allow us  to properly  define and state all the statistical properties of the system. 
We are able to construct such a measure and to prove its uniqueness. 
Moreover we show, under a few assumptions, the stochastic stability of the perturbed system, namely the weak convergence of the stationary measure to the invariant measure for the unimodal map in the zero noise limit. 
We also define an average Lyapunov exponent, still in terms of the stationary measure, as a sensible indicator of the slow motion, and prove its continuity with respect to the parameters defining the system.
We show that, depending on the parameters, the average Lyapunov exponent can be either negative or positive, leading to two qualitatively different (periodic- and chaotic-like) dynamics of the leverage.

We then estimate the parameters of the map via a method based on deep neural networks, whose efficiency was tested in a large testing set. 
Assuming  the  proposed  unimodal  map  with  heteroscedastic noise as data generating process for the banks leverage, we estimated the parameters on quarterly data of about 5,000 US Commercial Banks  via  the  proposes CNN architecture. 
By investigating the period  from  March  2001 to  December  2014,  for  a  total  of  59  quarters, we  found  that the parameters of a sizable fraction of banks lie in the dynamical core region of the parameter space and that the large banks' leverage tends to be more chaotic than the one of small ones. 
The latter finding was corroborated also by using a non-parametric approach.

We believe that the proposed methodologies may offer revealing perspectives for future works. 
For instance, it would be interesting to extend the employed mathematical techniques to study a model in which more than one asset and one bank are present in the system.

\section{Acknowledgements}

This research was supported by the research project `Dynamics and Information Research Institute - Quantum Information, Quantum Technologies' within the agreement between UniCredit Bank and Scuola Normale Superiore.
F.L.\ and S.M.\ acknowledge partial support by the European Program scheme `INFRAIA-01-2018-2019: Research and Innovation action', grant agreement \#871042 'SoBigData++: European Integrated Infrastructure for Social Mining and Big Data Analytics'.
S.V.\ is grateful to the Centro di Ricerca Matematica Ennio de Giorgi,
Scuola Normale Superiore,
Laboratoire International Associ\'e LIA LYSM, the INdAM (Italy) and the UMI-CNRS 3483, Laboratoire Fibonacci (Pisa), where this work has been initiated and completed under a CNRS delegation,
for the support.
S.V.\ thanks C.~Gonzalez-Tokman for useful discussion about Section~\ref{RLESection}.

\appendix
\section{The Chaos Decision Tree Algorithm}\label{app:cdta}
The Chaos Decision Tree Algorithm \cite{toker2020simple} is a non-parametric
chaos-detection tool which has been developed 
with the goal of being especially robust to 
measurement noise.
It provides an automated processing pipeline which has been showed to be able to detect the presence (or absence) of chaos in noisy recordings, even for difficult edge cases.
We use it in our work to identify periodic/chaotic time series without any reference to our model.
It is meant to provide an independent check on the 
existence of chaotic behaviour in leverage time series and to support the evidence that it may depend on the bank's size. 

The algorithm classifies a time-series as either stochastic or periodic, or chaotic. 
We now briefly explain how the algorithm works. 
The first step 
is to test if data are stochastic. 
This is done via surrogate-based approach by comparing the permutation entropy of the original time-series to the permutation entropy of random surrogates of that time-series by using a combination of Amplitude Adjusted Fourier Transform surrogates and Cyclic Phase Permutation surrogates. 
If the permutation entropy of the original time-series falls within either surrogate distribution, the time-series is classified as stochastic. 
If the permutation entropy falls outside the surrogate distribution, then the algorithm proceeds to de-noise the inputted signal by using the Schreiber's noise-reduction algorithm \cite{Schreiber93}. 
Notice that the calculation of the permutation entropy relies on two parameters: the permutation order and the time-lag.
The time-lag has been set to 1 as suggested in \cite{toker2020simple}. 
The choice of the order of the permutation is made in order to maximize the detection of chaotic series in our model.

To test the CDTA algorithm on our model, we first consider the deterministic map and generate $100$ chaotic and $100$ periodic time-series of length 59	from the dynamical core. 
We then apply CDTA to these series for different values of the permutation order ($\in\{3,\ldots,8\}$).
While the periodicity accuracy is maximized for a value equal to 8 ($92\%$ of the periodic series are correctly detected as such and the remaining $8\%$ are labeled chaotic), the chaos detection accuracy is maximized for a choice of the permutation order equal to 5 ($65\%$ of the series are correctly detected as chaotic and the remaining periodic). 
Because of the purpose of this paper, we fix the permutation order to 5, but we have checked that the conclusions of our data analysis (in particular Table~\ref{tab:iterate_k_all}) would have been the same with a different choice.

At this point, the algorithm checks for signal oversampling and, if the data are over-sampled, the algorithm iteratively down-samples the data until they are no longer over-sampled. 
Finally, CDTA performs the 0-1 chaos test \cite{gottwald09} on the input data. 
Ref.~\cite{toker2020simple} points out that the 0-1 chaos test has been modified from the original one to be less sensitive to noise.  
Then, it suppresses the correlations arising from quasi-periodicity, and normalizes the standard deviation of the test signal. 
The value for the parameter that suppresses signal correlations is chosen based on ROC analyses.  
The modified 0-1 test provides a single statistic, \textit{K}, which approaches 1 for chaotic systems and approaches 0 for periodic systems. 
The algorithm sets up a cutoff for \textit{K} based on the length of the time-series. 
If \textit{K} is greater than the cutoff, the data are classified as chaotic, and if they are smaller than or equal to the cutoff, they are classified as periodic.

\subsection{Simulations}\label{subsec:dta}

\begin{table}
\centering
\begin{tabular}{|c|c|c||c|c|c||c|c|c|}
\hline
\multirow{2}{*}{Iterate} & Series & \multirow{2}{*}{$\n$} & 
\multicolumn{3}{c|}{Dynamical core} & \multicolumn{3}{c|}{Not dynamical core} \\
\cline{4-9}
& length && S (\%) & P (\%) & C (\%) & S (\%) & P (\%) & C (\%) \\
\hline

\multirow{12}{*}{1} & \multirow{3}{*}{59} & 5 & 33.5 & 5.14 & \textbf{61.4} & \textbf{57.4} & 4.81 & 37.7 \\
&& 20 & 29.3 & 4.28 & \textbf{66.5} & \textbf{69.9} & 3.28 & 26.8 \\
&& 100 & 24.3 & 6.99 & \textbf{69.6} & \textbf{88.1} & 2.94 & 8.97 \\
\cline{2-9}
& \multirow{3}{*}{295} & 5 & 2.2 & 1.7 & \textbf{96.1} & 22.5 & 6.24 & \textbf{71.2} \\
&& 20 & 0.1 & 1.9 & \textbf{98} & 43.8 & 10.4 & \textbf{45.8} \\
&& 100 & 0 & 2.3 & \textbf{97.7} & \textbf{73.8} & 8.35 & 17.9 \\
\cline{2-9}
& \multirow{3}{*}{590} & 5 & 0 & 0.7 & \textbf{99.3} & 13.1 & 6.08 & \textbf{80.9} \\
&& 20 & 0 & 0.4 & \textbf{99.6} & 33.6 & 8.5 & \textbf{57.9} \\
&& 100 & 0 & 0.4 & \textbf{99.6} & \textbf{66.4} & 8.04 & 25.6 \\
\cline{2-9}
& \multirow{3}{*}{1180} & 5 & 0 & 0.1 & \textbf{99.9} & 10.9 & 3.44 & \textbf{85.6} \\
&& 20 & 0 & 0 & \textbf{100} & 27.7 & 5.57 & \textbf{66.8} \\
&& 100 & 0 & 0 & \textbf{100} & \textbf{60.2} & 5.29 & 34.5 \\
\hline
\hline
\multirow{12}{*}{2} & \multirow{3}{*}{59} & 5 & \textbf{75.7} & 2.17 & 22.2 & \textbf{83.8} & 1.33 & 14.9 \\
&& 20 & \textbf{80.1} & 1.65 & 18.2 & \textbf{92.8} & 0.26 & 6.91 \\
&& 100 & \textbf{86.6} & 1.43 & 11.9 & \textbf{96.6} & 0.56 & 2.81 \\
\cline{2-9}
& \multirow{3}{*}{295} & 5 & 39.4 & 0 & \textbf{60.6} & 40.5 & 3.26 & \textbf{56.2} \\
&& 20 & 38.6 & 0.6 & \textbf{60.8} & \textbf{70} & 3.7 & 26.3 \\
&& 100 & 21 & 1.2 & \textbf{77.8} & \textbf{83.9} & 2.96 & 13.1 \\
\cline{2-9}
& \multirow{3}{*}{590} & 5 & 27.6 & 0 & \textbf{72.4} & 25.6 & 3.23 & \textbf{71.2} \\
&& 20 & 10.6 & 0 & \textbf{89.4} & \textbf{52.3} & 4.24 & 43.4 \\
&& 100 & 4.8 & 0.6 & \textbf{94.6} & \textbf{74.7} & 2.88 & 22.4 \\
\cline{2-9}
& \multirow{3}{*}{1180} & 5 & 11 & 0 & \textbf{89} & 13.7 & 2.22 & \textbf{84.1} \\
&& 20 & 0.2 & 0 & \textbf{99.8} & 39.4 & 2.82 & \textbf{57.7} \\
&& 100 & 0.2 & 0 & \textbf{99.8} & \textbf{64.1} & 2.24 & 33.7 \\
\hline
\hline
\multirow{12}{*}{3} & \multirow{3}{*}{59} & 5 & \textbf{85.3} & 0.47 & 14.2 & \textbf{83.9} & 1.33 & 14.9 \\
&& 20 & \textbf{84.5} & 1.17 & 14.3 & \textbf{92.8} & 0.26 & 6.91 \\
&& 100 & \textbf{89.5} & 1.43 & 8.79 & \textbf{96.6} & 0.56 & 2.81 \\
\cline{2-9}
& \multirow{3}{*}{295} & 5 & 36 & 0 & \textbf{64} & 40.5 & 3.26 & \textbf{56.2} \\
&& 20 & 32.6 & 0.4 & \textbf{67} & \textbf{70} & 3.7 & 26.3 \\
&& 100 & 23 & 1.6 & \textbf{75.4} & \textbf{83.9} & 2.96 & 13.1 \\
\cline{2-9}
& \multirow{3}{*}{590} & 5 & 20.6 & 0 & \textbf{79.4} & 25.6 & 3.23 & \textbf{71.2} \\
&& 20 & 9.6 & 0 & \textbf{90.4} & \textbf{52.3} & 4.24 & 43.4 \\
&& 100 & 9.4 & 0.8 & \textbf{89.8} & \textbf{74.7} & 2.88 & 22.4 \\
\cline{2-9}
& \multirow{3}{*}{1180} & 5 & 6 & 0 & \textbf{94} & 13.7 & 2.22 & \textbf{84.1} \\
&& 20 & 2 & 0 & \textbf{98} & 39.4 & 2.82 & \textbf{57.7} \\
&& 100 & 2.8 & 0.8 & \textbf{96.4} & \textbf{64.1} & 2.24 & 33.7 \\
\hline
\end{tabular}
\caption{Percentage of time-series detected by the Chaos Decision Tree Algorithm as stochastic (S), periodic (P) and chaotic (C) in the dynamical core and its complement as a function of the time-series length and $\n$ for different iterates of the map.}
\label{tab:iterate_k_all}
\end{table}

In this section we present some numerical investigations showing how CDTA performs when simulating (noisy) time-series from the map described in Section~\ref{MainExampleSection}. 
Specifically, we run the following two numerical experiments:
\begin{enumerate}[label=(\roman*)]
    \item\label{itm:cdta_i} First, we simulate time-series from the dynamical core area (i.e.\ time-series for which the pairs $(\phi^*, \omega)$ satisfy condition~\ref{itm:C3}: $T(\Delta) < c < \Delta < 1$).
    More precisely, we simulate $500$ samples of different length and level of noise, which is captured by the variable $\n$. 
    \item\label{itm:cdta_ii} Second, we simulate time-series from outside the dynamical core area
    (the map $T$ thus satisfies \ref{itm:C1} or \ref{itm:C2}).
    The remaining simulation setting coincides with that in \ref{itm:cdta_i}.
\end{enumerate}
The procedures explained in \ref{itm:cdta_i}--\ref{itm:cdta_ii} are also repeated when $T$ is replaced by its $k$-th iterate $T^{k}$, $k=2, 3$; see discussion in Section~\ref{sec:estimationDNN}. 
Table~\ref{tab:iterate_k_all} collects the results.
We observe that when both $\n$ and the time series length are large, CDTA classifies almost all the time series in the dynamical core as chaotic, while those outside it are never classified as such. 
This is quite independent on $k$. 
By decreasing either $\n$ or the time series length, the classification is less precise and this effect is stronger for larger values of $k$. 
In the regime of length comparable with our empirical data ($N=59$), roughly a third (for $k=1$) or up to $85\%$ (for $k=3$) of the time series in the dynamical core are classified as stochastic, showing the limits of the CDTA when the time series are short and/or the noise is large.


\end{document}